\documentclass[final,onecolumn,12pt]{IEEEtran}
\usepackage{graphicx,psfrag,epsfig,epsf,latexsym,hhline,amsmath,amssymb,multirow}
\usepackage[usenames,dvipsnames]{pstricks}
\usepackage{pst-plot}
\usepackage{pstricks-add}
\usepackage{hyperref}
\usepackage{graphicx}
\usepackage{amsthm}
\usepackage{algorithm}
\usepackage{algpseudocode}
\usepackage[noadjust]{cite}
\usepackage{blindtext}
\usepackage{etoolbox,tcolorbox}
\input{Jerry.def}

\begin{document}
\title{Layered Space-Time Index Coding}
\author{Yu-Chih Huang, Yi Hong, Emanuele Viterbo, and Lakshmi Natarajan

\thanks{Y.-C. Huang is with the Department of Communication Engineering, National Taipei University, 237 Sanxia District, New Taipei City, Taiwan (email: ychuang@mail.ntpu.edu.tw).}

\thanks{Y. Hong and E. Viterbo are with the Department of Electrical and Computer System Engineering, Monash University, VIC 3800, Australia (e-mail: \{yi.hong, emanuele.viterbo\}@monash.edu).}

\thanks{L. Natarajan is with the Department of Electrical Engineering, Indian Institute of Technology Hyderabad, Sangareddy 502285, India (e-mail: lakshminatarajan@iith.ac.in).}
}

\maketitle

\begin{abstract}
Multicasting $K$ independent messages via multiple-input multiple-output (MIMO) channels to multiple users where each user already has a subset of messages as side information is studied. A general framework of constructing layered space-time index coding (LSTIC) from a large class of space-time block codes (STBC), including perfect STBC, is proposed. We analyze the proposed LSTIC and show that it provides minimum determinant gains that are exponential with the amount of information contained in the side information for {\it any possible side information}. When constructed over a perfect STBC, the proposed LSTIC is itself a perfect STBC and hence many desired properties are preserved. To illustrate, we construct LSTIC over the following well-known STBCs: Golden code; $3\times 3$, $4\times 4$, $6\times 6$ perfect STBCs; and Alamouti code. Simulation results show that the obtained side information gain can be well predicted by our analysis.
\end{abstract}


%------------------------------------------------------------------------

\section{Introduction}\label{sec:intro}
The index coding problem \cite{birk98,birk06} studies the problem of optimally broadcasting independent messages via noiseless links to multiple receivers where each receiver demands a subset of messages and already has another subset of messages as side information. The side information at a receiver is described by an index set and could be obtained from various means depending on the application. For example, in retransmissions in broadcast channel \cite{birk98}, the side information is decoded from the previous received signals; in the coded caching technique \cite{niesen14,niesen16}, the side information is prefetched into users' local cache memories during off-peak hours; and in wireless relay networks \cite{oechtering08,KramerShamai07,ccWang13}, the side information is the users' own data and/or is decoded/overheard from the previous sessions.

At the physical layer, the index coding problem can in fact be modeled as the noisy broadcast channel with receiver side information. This problem has recently been investigated from two different perspectives and most of the prior works can be categorized accordingly into two groups. The first one including \cite{oechtering08,KramerShamai07,wu07,tie09,sima14,LawrenceOng15} focuses on characterizing the capacity region of the AWGN broadcast channel with message side information. The capacity region of the two-user broadcast channel with receiver message side information has been completely characterized \cite{oechtering08,wu07}. However, since the number of possible index sets increases exponentially with the number of users in the network, the problem quickly becomes intractable as the number of users increases. As a result, the capacity region for the three-user case has not been fully characterized for some index sets \cite{tie09,sima14,LawrenceOng15} and our knowledge about more than three users is limited to some special cases \cite{tuncel06,natarajan17codedSI}.

The second category including \cite{Rajan16psk,viterbo15_index_QAM,viterbo14index,huang15_lic,huang17GCIC} considers designing codes/constellations that possess some desired properties in the finite dimension regime. The main objective is to design codes such that the probability of error will decrease by an amount that is proportional to the amount of information contained in the side information. In \cite{Rajan16psk}, Mahesh and Rajan consider the AWGN broadcast channel and assume that the transmitter knows all the index sets, i.e., the side information configuration is available at the transmitter. The scheme proposed in \cite{Rajan16psk} consists of a linear index coding followed by an algorithm that maps coded bits onto a phase shift keying (PSK) modulation. It is shown in \cite{Rajan16psk} that this scheme indeed can provide a reduction in probability of error proportional to the amount of side information.

Another line of research within this category (\cite{viterbo15_index_QAM,viterbo14index,huang15_lic,huang17GCIC}), which seamlessly scales to any number of users, considers the scenario where the transmitter is oblivious of the index sets. This enables to handle large numbers of users, when the index sets to feedback to the transmitter  require excessive resources and/or the complexity of designing the specific index code becomes excessive.  The objective then becomes designing coding schemes that are fair to every possible index set. As a starting point, only the multicasting case is considered in \cite{viterbo15_index_QAM,viterbo14index,huang15_lic,huang17GCIC} where all the receivers demand all the messages.

In \cite{viterbo15_index_QAM} and \cite{viterbo14index}, Natarajan \textit{et al.} study code design for the AWGN broadcast channel where minimum distance is one of the most crucial parameters to be maximized. They first propose a coding scheme in \cite{viterbo15_index_QAM} by partitioning multi-dimensional pulse amplitude modulation (PAM) into subsets via computer search for up to five messages with the message size up to $2^6$. Exploiting the algebraic structure induced by the Chinese remainder theorem (CRT), a novel coding scheme, lattice index coding, is then proposed in \cite{viterbo14index} to accommodate {\it any number of messages} with message sizes relatively prime to each other. Both the schemes in \cite{viterbo15_index_QAM} and \cite{viterbo14index} are shown to provide gains in minimum distance exponential with the rate of the side information, for any index set.

In \cite{huang15_lic}, Huang considers the same multicasting problem with message side information, where each link experiences a Rayleigh fading channel on top of the AWGN noise. It is well-known that in contrast to the AWGN channel, maximizing minimum distance alone is far from enough for the Rayleigh fading channel and the minimum product distance dominates the performance \cite{joseph96, oggier04, tse_viswananth}. The lattice index coding scheme proposed in \cite{huang15_lic} generalizes the idea of \cite{viterbo14index} from some famous principal ideal domains to any ring of algebraic integers. It is shown that codes thus constructed over rings of algebraic integers of totally real number fields provide gains in minimum product distance that is exponential with the rate of the side information for any index set. The multicasting problem with message side information is then considered in \cite{huang17GCIC} under the $2\times 2$ MIMO setting where the transmitter and the receivers are equipped with two antennas. For such a MIMO setting, the minimum determinant of the code serves as one of the most important parameters to be maximized \cite{oggier07, tse_viswananth} and algebraic space-time block codes (STBC) constructed from cyclic division algebras \cite{alamouti98,sethuraman03stbc,belfiore05GC,Oggier06perfect} are a class of codes that possess many desired properties. Since CRT does not hold for non-commutative rings such as cyclic division algebras, the trick used in \cite{viterbo14index} and \cite{huang15_lic} does not work here in general. In \cite{huang17GCIC}, the problem is circumvented by using the bijective mapping between the Golden algebra and a commutative ring found in \cite{viterbo05partitionGC} together with some special ideals whose group structure is preserved by the mapping. As a result, we successfully construct Golden-coded index coding from Golden code, a subclass of perfect codes for the case with two transmitter and receiver antennas, and show that the minimum determinant increases exponentially with the rate of the side information for any index set. %However, this technique is not easy to be generalized to other STBC and with this technique the possible message sizes are going to be limited to some special prime powers.

\subsection{Contributions}
In this work, we consider the problem of multicasting over MIMO channel with message side information. We propose {\em layered space-time index coding} (LSTIC), a general framework of constructing lattice space-time index codes from algebraic STBC. We exploit the algebraic structure of these codes to encode the different messages into subcodes, which preserve all the good properties of the STBC, such as non-vanishing determinant and power efficiency. %It has the property that every LSTIC is itself a STBC that the LSTIC is constructed over. Therefore, when constructed over a perfect STBC, the proposed LSTIC is guaranteed to be a perfect STBC.

Any receiver that has some of the messages as side information will be decoding a subcode that has an improved performance in terms of error probability. We provide a lower bound on the {\em side information gain} for any side information configuration. The side information gain essentially measures the SNR reduction (normalized by the rate of the side information)  to achieve the same error probability, given the side information. This lower bound implies an exponential increase of minimum determinant and is universal in the sense that it holds for {\it any possible index set}.

We apply the proposed framework with the Golden code, $3\times 3$, $4\times 4$, $6\times 6$ perfect STBCs, and Alamouti code, and show that our analysis well predicts the actual side information gains obtained from simulations. For each of the above codes, we also provide a table of the corresponding prime ideal factorizations for $p<100$, over which the LSTIC can be constructed according to the desired message sizes.

We note that the technique used in \cite{huang17GCIC} requires the code to be constructed over prime ideals whose group structure are preserved by the bijective mapping of \cite{viterbo05partitionGC} and thereby limits the possible message sizes. In contrast to the Golden-coded index coding in \cite{huang17GCIC} working only for Golden code with the message sizes confined to some particular prime powers, the proposed LSTIC is quite general in the following senses: 1) it works for a large class of STBC constructed from cyclic division algebras; and 2) it has less restriction on the message size. %In fact, for any prime $p$, the LSTIC can have a message whose size is $p^f$ for some small integer $f$ that can be directly determined from the algebraic properties.
We would like to emphasize here that when specialized to the Golden code, the proposed LSTIC is not a special case of the Golden-coded index coding in \cite{huang17GCIC} and vice versa, due to the different message sizes.

%Tables of prime ideal factorizations for $p<100$ for the LSTIC constructed from the above perfect STBC are also provided from which one can choose the ideals resulting in the desired message sizes.

\subsection{Notations}
Throughout the paper, the following notations are used. Matrices are written in capital boldface, for example $\mathbf{X}$. Let $\msf{i}\defeq \sqrt{-1}$ and $\omega\defeq e^{\msf{i} 2\pi/3}$ be the primitive cube root of unity. We denote by $\Z$, $\Zi \defeq \{a+b\msf{i}|a,b\in\Z\}$, and $\Zw \defeq \{a + b\omega|a,b\in\Z\}$ the ring of integers, the ring of Gaussian integers, and the ring of Eisenstein integers, respectively. Also, we denote by $\Q$, $\mbb{R}$, and $\mbb{C}$ the field of rational numbers, the field of real numbers, and the field of complex numbers, respectively.

\subsection{Organization}
The rest of the paper is organized as follows. In Section~\ref{sec:problem}, we state the problem of physical-layer index coding over MIMO channel and formally define the side information gain, the performance measure that we will use throughout the paper. Background knowledge on algebra, algebraic number theory, and cyclic division algebra is given in Section~\ref{sec:background}. The LSTIC is then proposed and analyzed in Section~\ref{sec:LSTIC}. In Sections~\ref{sec:L_GIC}-\ref{sec:AlaIC}, we construct LSTIC over Golden code, $3\times 3$ perfect STBC, $4\times 4$ perfect STBC, $6\times 6$ perfect STBC, and Alamouti code. We then conclude the paper in Section~\ref{sec:conclude}.

\section{Problem Statement}\label{sec:problem}
Consider the network shown in Fig.~\ref{fig:system_model} where there is a base station broadcasting messages to $L$ users. The base station is equipped with $n_t$ antennas and each user is equipped with $n_r$ antennas. There are $K$ independent messages $\{w_1,\ldots,w_K\}$ collocated at the base station and each $w_k$ is uniformly distributed over $\{1,\ldots,M_k\}$. Each user demands all the $K$ messages and already has a subset of the messages as side information. For user $\ell$, we denote by $\mc{S}_{\ell}\subseteq\{1,\ldots,K\}$ the index set and the side information at the user is $w_{\mc{S}_{\ell}}\defeq \{w_s | s\in\mc{S}_{\ell}\}$. The base station encodes the messages across space ($n_t$ antennas) and time ($T$ symbol durations) into an $n_t\times T$ codeword matrix $\mbf{X}$ where each entry $x_{jt}\in\mbb{C}$ and the codeword is subject to the power constraint $\mbb{E}[\|\mathbf{X}\|^2]=n_tT$. In a space-time code, each codeword $\mathbf{X}$ is used to transmit $r$ information-bearing real symbols. We denote by $R_k=\log_2(M_k)/r$ the rate of the message $w_k$ measured in bits per real symbol. The signal model between the base station and the user $\ell$ is given by
\begin{equation*}
    \mathbf{Y}_{\ell} = \mathbf{H}_{\ell}\mathbf{X} + \mathbf{Z}_{\ell},
\end{equation*}
where $\mathbf{Y}_{\ell}$ is of size $n_r\times T$, $\mbf{H}_{\ell}$ is a random $n_r\times n_t$ matrix with each element i.i.d. $\sim\mc{CN}(0,1)$, and $\mbf{Z}_{\ell}$ is a random $n_r\times T$ matrix with each element i.i.d. $\sim\mc{CN}(0,\sigma^2_l)$. Each user is assumed to know the channel matrix $\mbf{H}_{\ell}$ associated with its received signal, i.e., channel state information at the receiver is assumed. The signal-to-noise power ratio (SNR) is defined as $\mathrm{SNR}_l\defeq \frac{n_t}{\sigma^2_l}$.
\begin{figure}
    \centering
    \includegraphics[width=3.5in]{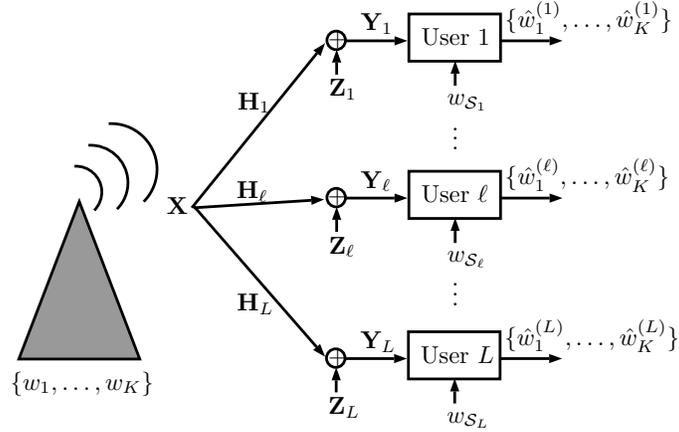}
    \caption{Multicasting $\{w_1,\ldots,w_{K}\}$ over MIMO channel to $L$ receivers where each receiver $\ell\in\{1,\ldots,L\}$ has a subset of messages $w_{\mc{S}_{\ell}}$ as side information.}
    \label{fig:system_model}
\end{figure}

Let $\phi$ be a bijective encoding function that maps the messages $(w_1,\ldots,w_K)$ to the transmitted signal $\mbf{X}$. The codebook $\mc{C}$ is the collection of codewords given by
\begin{equation*}
    \mc{C} = \left\{\mbf{X}=\phi(w_1,\ldots,w_K)| w_k\in\{1,\ldots,M_k\}, \forall k\right\}.
\end{equation*}
Based on the received signal $\mbf{Y}_{\ell}$ and side information $w_{\mc{S}_{\ell}}$, the user ${\ell}$ forms $\{\hat{w}_1^{(\ell)},\ldots,\hat{w}_K^{(\ell)}\}$ (or equivalently $\hat{\mbf{X}}^{(\ell)}$) an estimate of $\{w_1,\ldots,w_K\}$ (or equivalently $\mbf{X}$). The probability of error is defined as
\begin{align*}
    p_e^{(\ell)} &\defeq \mathrm{Pr}\{ \{w_1,\ldots,w_K\} \neq \{\hat{w}_1^{(\ell)},\ldots,\hat{w}_K^{(\ell)}\}\} \nn \\
    &=\mathrm{Pr}\{ \mbf{X}\neq\hat{\mbf{X}}^{(\ell)}\},
\end{align*}
where the second expression is often called the codeword error rate (CER). We emphasize here that the index set $\mc{S}_{\ell}$ can be any subset of $\{1,\ldots,K\}$ and is oblivious to the base station. This makes the problem of every $\ell$ identical for the base station. We therefore focus on a generic user and drop the subscript (superscript in some cases) $\ell$. The dummy variable $\ell$ is then released for later use.

%Let $\mc{C}$ be the transmitted codebook and the encoder maps the messages to the codewords as $f(w_1,\ldots,w_K)=\mathbf{X}\in \mc{C}$.
Following \cite{oggier07}, we define $\mathbf{A}\defeq (\mathbf{X}-\mathbf{X}')(\mathbf{X}-\mathbf{X}')^{\dag}$ for any pair of codeword matrices $\mathbf{X},\mathbf{X}'\in\mc{C}$. Let $r$ be the rank of $\mathbf{A}$. For the generic user with $\mc{S}=\emptyset$, in the high SNR regime, the probability of mistaking $\mbf{X}'$ for $\mbf{X}$ can be bounded as
\begin{equation*}%\label{eqn:pairwise_pe}
    \mathrm{Pr}(\mathbf{X}\rightarrow\mathbf{X}') \leq \left(\frac{\mathrm{SNR}\Delta^{1/r}}{4 n_t}\right)^{-rn_r},
\end{equation*}
where $\Delta =\prod_{m=1}^r\lambda_m$ with $\lambda_1,\ldots,\lambda_m$ being the non-zero eigenvalues of $\mathbf{A}$. Moreover, for full rank codes, i.e., $r=n_t$ and
\begin{equation*}
    \Delta =\prod_{m=1}^{n_t}\lambda_m = \det(\mathbf{A})\neq 0,
\end{equation*}
we define the {\em minimum determinant} of $\mc{C}$ as follows,
\begin{equation*}%\label{eqn:min_det_C}
    \delta(\mc{C}) \defeq \min_{\mathbf{X}\neq\mathbf{X'}\in\mc{C}}\det(\mathbf{A}).
\end{equation*}
If $\mc{C}$ is carved from a lattice $\Lambda$ \cite{conway1999sphere}, we have
\begin{equation}\label{eqn:min_det_C_lattice}
    \delta(\mc{C}) = \min_{\mathbf{X}\neq\mathbf{0}\in \Lambda} \det(\mathbf{X})^2.
\end{equation}
%The probability of error of a code carved from a lattice can then be approximated as
%\begin{align}\label{eqn:union_bound}
%    p_e &\leq \frac{1}{|\mc{C}|}\sum_{\mathbf{X}\in\mc{C}}\sum_{\mathbf{X}'\neq\mathbf{X}}\Pp(\mathbf{X}\rightarrow\mathbf{X}') \nonumber \\
%    &\leq \frac{1}{|\mc{C}|} \sum_{\mathbf{X}\neq\mathbf{0}} \left(\frac{\mathrm{SNR} \delta^{1/n_t}}{4 n_t}\right)^{-n_t n_r}
%\end{align}
To estimate the probability of error more accurately, let us define $N_{\mathbf{X}}$ the number of codewords $\mathbf{X'}\in\mc{C}$ resulting in $\det(\mbf{A})=\delta(\mc{C})$ and define
\begin{equation}\label{eqn:N_avg}
    N_{\mc{C}}\defeq \frac{1}{|\mc{C}|}\sum_{\mathbf{X}\in\mc{C}} N_{\mathbf{X}},
\end{equation}
the average of $N_{\mbf{X}}$ over $\mbf{X}\in\mc{C}$. For a STBC carved from a lattice, we can now approximate the probability of error as
\begin{align}\label{eqn:union_bound}
    p_e &= \frac{1}{|\mc{C}|}\sum_{\mathbf{X}\in\mc{C}} \Pp\left(\bigcup_{\mathbf{X}'\neq\mathbf{X}} \mathbf{X}\rightarrow\mathbf{X}'\right) \nonumber \\
    &\overset{(a)}{\approx} \frac{1}{|\mc{C}|} \sum_{\mathbf{X} \in\mc{C}} N_{\mathbf{X}}\left(\frac{\mathrm{SNR} \delta(\mc{C})^{1/n_t}}{4 n_t}\right)^{-n_t n_r} \nonumber \\
    &= N_{\mc{C}}\left(\frac{\mathrm{SNR} \delta(\mc{C})^{1/n_t}}{4 n_t}\right)^{-n_t n_r},
\end{align}
where the approximation in ($a$) will become quite accurate in the high SNR regime.

Having had the approximation in \eqref{eqn:union_bound}, we can now follow \cite{huang17GCIC} to derive the side information gain as follows. We first note that with the knowledge of side information $w_s = v_s, \forall s\in\mc{S}$, the generic user can throw away all the codewords that do not correspond to this side information. The codebook then becomes
\begin{equation*}
    \mc{C}_{\mc{S}}\defeq \left\{\mbf{X}=\phi(d_1,\ldots,d_K)\left|
                                                              \begin{array}{ll}
                                                              d_k = v_k  , & \hbox{$k\in\mc{S}$;} \\
                                                               d_k\in \{1,\ldots,M_k\}  , & \hbox{otherwise.}
                                                              \end{array}
                                                             \right.\right\},
\end{equation*}
a subcode of $\mc{C}$. Since $\mc{C}_{\mc{S}}\subseteq \mc{C}$, the minimum determinant of $\mc{C}_{\mc{S}}$, $\delta(\mc{C}_{\mc{S}})$, will be no less than $\delta(\mc{C})$. Let us now see how gains in minimum determinant can be translated into SNR gains. Following \cite{huang17GCIC}, we let $\SNR$ and $\SNR_{\mc{S}}$ be the SNR required for the codebooks $\mc{C}$ and $\mc{C}_{\mc{S}}$, respectively, to achieve a same error probability $p_e$. Then \eqref{eqn:union_bound} says that
\begin{align}
    &N_{\mc{C}}\left(\frac{\mathrm{SNR} \delta(\mc{C})^{1/n_t}}{4 n_t}\right)^{-n_t n_r} \approx N_{\mc{C}_{\mc{S}}}\left(\frac{\mathrm{SNR}_{\mc{S}} \delta(\mc{C}_{\mc{S}})^{1/n_t}}{4 n_t}\right)^{-n_t n_r} \nonumber \\
    &(\Leftrightarrow)\quad 10\log_{10}(\mathrm{SNR}) - 10\log_{10}(\mathrm{SNR}_{\mc{S}}) \approx \nonumber \\
    &\hspace{1.5cm} \frac{1}{n_tn_r}10\log_{10}\left(\frac{N_{\mc{C}}}{N_{\mc{C}_{\mc{S}}}}\right) + \frac{1}{n_t}10\log_{10}\left(\frac{\delta(\mc{C}_{\mc{S}})}{\delta(\mc{C})}\right), \label{eqn:snr_gain_exact}
\end{align}
which represents the SNR gain in dB provided by the side information $w_{\mc{S}}$. As mentioned in \cite{huang17GCIC} and many other work in the space-time code literature, it is in general not an easy task to keep tracking both $N_{\mc{C}_{\mc{S}}}$ and $\delta(\mc{C}_{\mc{S}})$ for lattice codes; we thereby focus solely on $\delta(\mc{C}_{\mc{S}})$ as our design guideline and define the SNR gain as $10\log_{10}\left(\delta(\mc{C}_{\mc{S}})/\delta(\mc{C}) \right)^{n_t}$ dB. To get a fair comparison for every possible side information, we then normalize this side information gain by the rate of the side information and define the normalized side information gain as
\begin{equation}\label{eqn:side_info_gain}
    \Gamma(\mc{C},\mc{S})\defeq \frac{10\log_{10}\left(\frac{\delta(\mc{C}_{\mc{S}})}{\delta(\mc{C})}\right)}{n_t R_{\mc{S}}},
\end{equation}
where the rate of the side information is defined as $R_{\mc{S}}\defeq \sum_{s\in\mc{S}} R_s$ and is measured in bits per real symbol, which makes the normalized side information gain having the unit ``dB/bits per real symbol". The side information gain essentially serves as an approximation of the SNR gain provided by side information $w_{\mc{S}}$, normalized by the rate of $w_{\mc{S}}$. We note that involving the first term of \eqref{eqn:snr_gain_exact} into the definition of side information gain results in a better approximation. Hence, although we use \eqref{eqn:side_info_gain} as the design guideline throughout the paper, \eqref{eqn:snr_gain_exact} is also used to confirm the simulation results.

\section{Background}\label{sec:background}
In this section, we first review basic knowledge including algebra and algebraic number theory. We then focus on cyclic division algebra and its connection to lattice STBC. To make the paper concise, we only review the minimum required background for understanding the discussion that follows. For details, please refer, for example, to \cite{Hungerford74,oggier04,oggier07,StewardTall,lang94}.

\subsection{Algebra}
Let $\mc{R}$ be a {\em commutative ring} equipped with two operations addition $+$ and multiplication $\cdot$. %A subring $\mc{T}$ of $\mc{R}$ is a subset of $\mc{R}$ that is itself a ring under operation $+$ and $\cdot$ restricted to $\mc{R}$. Note that in our definition, $\mc{T}$ does not necessarily have a multiplicative identity.
%Let $a, b\neq 0 \in\mc{R}$ but $ab = 0$, then $a$ and $b$ are \textit{zero divisors}. If $ab = ba = 1$, then we say $a$ is a \textit{unit}. Two elements $a, b\in\mc{R}$ are associates if $a$ can be written as the multiplication of a unit and $b$. A non-unit element $\tau\in\mc{R}$ is a prime if whenever $\tau$ divides $ab$ for some $a, b \in \mc{R}$, either $\tau$ divides $a$ or $\tau$ divides $b$. An \textit{integral domain} is a commutative ring with identity and no zero divisors.
An {\em ideal} $\mfk{I}$ of $\mc{R}$ is an additive subgroup of $\mc{R}$ with respect to $+$ that absorbs the multiplication of $\mc{R}$, i.e., it satisfies $a\cdot r\in\mfk{I}$ for $a\in\mfk{I}$ and $r\in\mc{R}$. An ideal $\mfk{I}$ is a {\em principal ideal} if it can be generated by a singleton, i.e., $\mfk{I}=a\mc{R}$ for some $a\in\mc{R}$. A {\em proper ideal} $\mfk{I}$ is an ideal that is at the same time, a proper subset of $\mc{R}$, i.e., $\emptyset\neq\mfk{I}\subset\mc{R}$. %A proper ideal is necessarily a subring\footnote{This is not true in general if subrings are defined with the additional requirement of having multiplicative identity, i.e., being an unitary ring.}.

%A \textit{principal ideal domain} (PID) is an integral domain in which every ideal is principal. Famous and important examples of PID include $\mbb{Z}$, $\Zi$ and $\Zw$.
For an ideal $\mfk{I}$ and any two elements $a, b\in\mc{R}$, $a$ is congruent to $b$ \textit{modulo} $\mfk{I}$ if and only if $a-b\in\mfk{I}$, which defines an equivalence relation. The {\em quotient ring} $\mc{R}/\mfk{I}$ of $\mc{R}$ by $\mfk{I}$ is the collection of equivalence classes with addition and multiplication defined as the original ones followed by modulo $\mfk{I}$ operation as follows,
\begin{align*}
    (a+\mfk{I})+(b+\mfk{I}) &= (a+b)+\mfk{I}, \text{~and} \\
    (a+\mfk{I})\cdot(b+\mfk{I}) &= (a\cdot b)+\mfk{I},
\end{align*}
respectively. A {\em prime ideal} $\mfk{p}$ of $\mc{R}$ is a proper ideal satisfying that whenever $ab\in\mfk{p}$ for $a, b\in\mc{R}$, then either $a\in\mfk{p}$ or $b\in\mfk{p}$. %A proper ideal $\mfk{I}$ of $\mc{R}$ is said to be a \textit{maximal ideal} if $\mfk{I}$ is not contained in any strictly larger proper ideal. It should be noted that every maximal ideal is also a prime ideal but the reverse may not be true. Let $\mc{R}_1, \mc{R}_2, \ldots, \mc{R}_L$ be a family of rings, the direct product of these rings, denoted by $\mc{R}_1\times \mc{R}_2\times \ldots \times\mc{R}_L$, is the direct product of the additive abelian groups $\mc{R}_l$ equipped with multiplication defined by the \textit{componentwise} multiplication.
We now define the {\em sum and product of ideals}. Let $\mfk{I}_1$ and $\mfk{I}_2$ be two ideals of $\mc{R}$, the sum of two ideals is itself an ideal and is defined as
\begin{equation*}
    \mfk{I}_1+\mfk{I}_2 \defeq \left\{a+b: a\in\mfk{I}_1, b\in\mfk{I}_2\right\}.
\end{equation*}
The product of $\mfk{I}_1$ and $\mfk{I}_2$ is again an ideal and is defined as
\begin{equation*}
    \mfk{I}_1\mfk{I}_2 \defeq \left\{ \sum_{i=1}^{n} a_i b_i: a_i\in\mfk{I}_1, b_i\in\mfk{I}_2, n\in\mbb{N} \right\}.
\end{equation*}
In general, $\mfk{I}_1\mfk{I}_2 \subseteq \mfk{I}_1 \cap \mfk{I}_2$. Two ideals are said to be \textit{relatively prime} if $\mc{R} = \mfk{I}_1+\mfk{I}_2$. When $\mfk{I}_1$ and $\mfk{I}_2$ are relatively prime, we further have $\mfk{I}_1\mfk{I}_2 = \mfk{I}_1\cap\mfk{I}_2$. We say $\mfk{I}_1$ divides $\mfk{I}_2$, denoted as $\mfk{I}_1|\mfk{I}_2$, if $\mfk{I}_2=\mfk{I}_1\mfk{I_3}$ for some ideal $\mfk{I}_3$ and consequently $\mfk{I}_2\subseteq \mfk{I}_1$.

%Two ideals are \textit{relatively prime} if
%\begin{equation}
%    \mc{R} = \mfk{I}_1+\mfk{I}_2 \triangleq \{a+b: a\in\mfk{I}_1, b\in\mfk{I}_2\}.
%\end{equation}

Consider two commutative rings $\mc{R}_1$ and $\mc{R}_2$ with two operations $(+,\cdot)$ and $(\oplus,\odot)$, respectively. A {\em ring homomorphism} between $\mc{R}_1$ and $\mc{R}_2$ is a function $\sigma:\mc{R}_1\rightarrow \mc{R}_2$ such that
\begin{align*}
    \sigma(a + b) &= \sigma(a) \oplus \sigma(b),~\forall a,b\in\mc{R}_1, \\
    \sigma(a\cdot b) &= \sigma(a)\odot \sigma(b),~\forall a,b\in\mc{R}_1.
\end{align*}
In other words, a ring homomorphism preserves the ring structure. A homomorphism is a monomorphism if it is injective and is an isomorphism if it is bijective. Moreover, an isomorphism $\sigma:\mc{R}_1\rightarrow \mc{R}_1$ is called automorphism.

We now review two classical results in ring theory whose proofs can be found in a standard textbook.
\begin{lemma}[{Second isomorphism theorem \cite[Theorem 2.12]{Hungerford74}}]
Let $\mc{R}$ be a commutative ring, $\mfk{I}_1$ and $\mfk{I}_2$ be two ideals. We have the following isomorphism,
\begin{equation*}
    \mfk{I}_1/(\mfk{I}_1\cap\mfk{I}_2) \cong (\mfk{I}_1+\mfk{I}_2)/\mfk{I}_2.
\end{equation*}
\end{lemma}
In fact, the second isomorphism theorem holds for the more general case where $\mfk{I}_1$ is only a subring and not necessarily an ideal.

\begin{lemma}[{Chinese remainder theorem \cite[Corollary 2.27]{Hungerford74}}]
Let $\mfk{I}_1,\ldots,\mfk{I}_n$ be ideals of a commutative ring $\mc{R}$. Moreover, $\mfk{I}_1,\ldots,\mfk{I}_n$ are relatively prime. We have
\begin{equation*}
    \mc{R}/\Pi_{i=1}^n\mfk{I}_i \cong \mc{R}/\mfk{I}_1\times\ldots\times\mc{R}/\mfk{I}_n.
\end{equation*}
where $\times$ stands for Cartesian product and the operations of the right hand side are defined componentwise.
\end{lemma}
%One property that will be frequently used is that when $\mfk{I}_1,\ldots,\mfk{I}_n$ are relatively prime, $\cap_{i=1}^n\mfk{I}_i$ becomes $\Pi_{i=1}^n\mfk{I}_i$ and CRT becomes
%\begin{equation}
%    \mc{R}/\Pi_{i=1}^n\mfk{I}_i \cong \mc{R}/\mfk{I}_1\times\ldots\times\mc{R}/\mfk{I}_n.
%\end{equation}

We provide a quick example for what have been reviewed above.
\begin{example}
    Consider $\mbb{Z}$ the set of all integers with ordinary addition $+$ and multiplication $\cdot$. Clearly, it forms a commutative ring. $2\mbb{Z}$ is the principal ideal of $\mbb{Z}$ consisting of all the even integers. Moreover, it is a prime ideal. The quotient $\mbb{Z}/2\mbb{Z} = \mbb{Z}_2$ forms a ring with addition $+\mod 2\mbb{Z}$ and multiplication $\cdot\mod 2\mbb{Z}$. Also, for $3\mbb{Z}$ another principal prime ideal of $\mbb{Z}$, we have the quotient ring $\mbb{Z}/3\mbb{Z} = \mbb{Z}_3$. Since $2\cdot (-1) + 3\cdot 1=1$, $2\mbb{Z}+3\mbb{Z} = \mbb{Z}$ and thus $2\mbb{Z}$ and $3\mbb{Z}$ are relatively prime. One can easily verify that $2\mbb{Z}\cap 3\mbb{Z}$ is precisely $6\mbb{Z}$. Now, the CRT guarantees the existence of a ring isomorphism between $\mbb{Z}_6=\mbb{Z}/6\mbb{Z}$ and $\mbb{Z}_2 \times \mbb{Z}_3$. One can verify that $\mc{M}(v_1,v_2) = 3v_1 - 2v_2\mod 6\mbb{Z}$ where $v_1\in\mbb{Z}_2$ and $v_2\in\mbb{Z}_3$ is a ring isomorphism.
\end{example}

\subsection{Algebraic Numbers and Algebraic Integers}
An {\em algebraic number} is a complex number that is a root of some polynomial with coefficients in $\mbb{Z}$. Let $\LL$ be a field and $\KK\subset\LL$ be a subfield; $\LL$ is said to be a field extension of $\KK$, which is usually denoted as $\LL/\KK$. $\LL$ can be viewed as a vector space over $\KK$. The {\em degree} of $\LL$ over $\KK$, denoted by $[\LL:\KK]$, is defined as the dimension of the vector space $\LL$ over $\KK$. A {\em number field} is a field extension of $\Q$ with finite degree, i.e., a finite extension $\KK/\Q$. Every number field $\KK$ can be generated from $\Q$ by adjoining an algebraic number $\theta$, i.e., $\KK=\Q(\theta)$.
%We define a number field to be a subfield $\mbb{K}$ of $\mbb{A}$ (hence a subfield of $\mbb{C}$) such that the degree $[\mbb{K}:\mbb{Q}]$ is finite.
%Theorem 2.2 in \cite{StewardTall} shows that any such $\mbb{K}$ is equal to $\mbb{Q}(\theta)$, the smallest subfield containing $\mbb{Q}$ and $\theta$, for some algebraic number $\theta$.
An {\em algebraic integer} is a complex number that is a root of some polynomial with the leading coefficient 1 and other coefficients in $\mbb{Z}$. For a number field $\KK$, we denote by $\Ok$ the {\em ring of integers} of $\KK$ which comprises all the algebraic integers in $\KK$.

Let $\LL/\KK$ be a {\em field extension} of $\KK$ with degree $[\LL:\KK]=n$. Throughout the paper, we will further assume that $\LL/\KK$ is a {\em Galois extension}. There are exactly $n$ distinct $\KK$-automorphisms $\sigma_i:\mbb{L}\rightarrow\mbb{L}$ for $i\in\{1,\ldots,n\}$, i.e., automorphisms that fix $\KK$. Such automorphisms are called (relative) {\em embeddings}. It can be shown that $\text{Gal}(\LL/\KK)\defeq\{\sigma_1,\ldots,\sigma_n\}$ form a group under function composition, which is called the {\em Galois group}. For $\alpha\in\LL$, we define the norm of $\alpha$ as
\begin{equation*}
    N_{\LL/\KK}(\alpha)=\prod_{i=1}^n\sigma_i(\alpha),
\end{equation*}
where $\sigma_2(\alpha),\ldots,\sigma_n(\alpha)$ are called the conjugates of $\sigma_1(\alpha)=\alpha$. Let $\{\alpha_1,\ldots,\alpha_n\}$ be an {\em integral basis} for $\OL$, such that any element in $\OL$ can be uniquely written as a linear combination of the basis element with coefficients $\Z$. The {\em discriminant} of a number field $\LL$ is defined as
\begin{equation*}
    d_{\LL}\defeq \det\left(
                                                 \begin{array}{cccc}
                                                   \sigma_1(\alpha_1) & \sigma_1(\alpha_2) & \ldots & \sigma_1(\alpha_n) \\
                                                   \sigma_2(\alpha_1) & \sigma_2(\alpha_2) & \ldots & \sigma_2(\alpha_n) \\
                                                   \vdots & \vdots & \ddots & \vdots \\
                                                   \sigma_n(\alpha_1) & \sigma_n(\alpha_2) & \ldots & \sigma_n(\alpha_n) \\
                                                 \end{array}
                                               \right)^2.
\end{equation*}

Let $\mfk{I}$ be an ideal in $\OL$, then $\mfk{I}$ can be generated by at most two elements, i.e., $\mfk{I}=\alpha\OL + \beta\OL$ for some $\alpha,\beta\in\OL$.
%The relative norm of $\mfk{I}$ is
%\begin{equation}
%    \mc{N}_{\LL/\KK}(\mfk{I}) \defeq \Ok \cap \prod_{\sigma\in\text{Gal}(\LL/\KK)} \sigma(\mfk{I}),
%\end{equation}
%and is itself an ideal.
The norm of $\mfk{I}$ is defined as
\begin{equation*}
    N(\mfk{I}) \defeq |\OL/\mfk{I}|.
\end{equation*}
Moreover, if $\mfk{I}=\alpha\OL$ is principal, $N(\mfk{I}) = |N_{\LL/\Q}(\alpha)|$.

Let $\mfk{p}$ be a prime ideal in $\OL$, the ring of integers of $\LL$ with $[\LL:\Q]=n$. We say that $\mfk{p}$ lies above a prime $p$ if $\mfk{p}\cap\mbb{Z}=p\mbb{Z}$. For a prime $p$, the principal ideal $p\OL$ can be factorized into $1\leq g\leq n$ prime ideals as
\begin{equation*}
    p\OL = \mfk{p}_1^{e_1}\cdot\ldots\cdot \mfk{p}_g^{e_g},
\end{equation*}
where $e_{i}, i\in\{1,\ldots,g\}$, is the {\em ramification index} of $\mfk{p}_i$. Also, for each $\mfk{p}_i$, we have $N(\mfk{p}_i)=p^{f_i}$ and $\OL/\mfk{p}_i\cong \mbb{F}_{p^{f_i}}$ where $1\leq f_i\leq n$ is the {\em inertial degree}. Overall, it can be shown that $\sum_{i=1}^g e_i f_i=n$. For a Galois extension, we have $e_1=e_2=\ldots=e_g=e$ and $f_1=f_2=\ldots=f_g=f$, which implies that $efg=n$. A prime $p$ is ramified in $\OL$ if not all $e_i=1$ in the factorization of $p\OL$. Ramified primes in $\OL$ are precisely those $p$ that divides the discriminant $d_{\LL}$.

\begin{example}
    Consider $\mbb{Q}(\msf{i})$ the field extension obtained from $\Q$ by adjoining $\msf{i}$. Every element in $\Q(\msf{i})$ has the form $a+b\msf{i}$ where $a,b\in\Q$; thus, it is a number field with degree 2. The two $\Q$-automorphisms are $\sigma_1(a+b\msf{i})\rightarrow a+b\msf{i}$ and $\sigma_2(a+b\msf{i})\rightarrow a-b\msf{i}$. The Galois group is cyclic and can be generated by $\sigma_2$. Since $\sigma_1$ is the identity mapping and $\sigma_2$ sends an element to its complex conjugate, the norm defined in this number field coincides with the Euclidean norm. The ring of integers is $\Zi$, the Gaussian integers, having integral basis $\{1,\msf{i}\}$. The discriminant is computed as follows,
\begin{equation*}
    d_{\Q(\msf{i})} = \det \begin{pmatrix}
                             1 & \msf{i} \\
                             1 & -\msf{i} \\
                           \end{pmatrix}^2 = -4.
\end{equation*}
Since $2|d_{\Q(\msf{i})}$, $2\Zi = \mfk{p}^2$ ramifies where $\mfk{p}=(1+\msf{i})\Zi$. This is the only ramified prime in $\mbb{Q}(\msf{i})$. Also, $5\Zi=\mfk{p}_1\mfk{p_2}$ splits into two prime ideals $\mfk{p}_1=(1+2\msf{i})\Zi$ and $\mfk{p}_2=(1-2\msf{i})\Zi$ with $e=1$ and $f=1$. Another example is that $3\Zi$ is itself a prime ideal with $e=1$ and $f=2$. In each case, we have $efg = 2$.
\end{example}

\subsection{Cyclic division algebra and lattice space-time codes}\label{sec:CDA_review}
An {\em algebra} $\mc{A}$ over a field $\LL$ is a set satisfying: {\em i)} it is a vector space over $\LL$; {\em ii)} it is a ring with respect to addition and multiplication by elements of $\mc{A}$; and {\em iii)} $(\alpha a) b =  a (\alpha b) = \alpha (a b)$ for any $\alpha\in\LL$ and $a,b\in\mc{A}$. Let $\LL/\KK$ be a field extension of $\KK$ of degree $n$ whose Galois group is a cyclic group generated by $\sigma$. One can construct a {\em cyclic algebra} $\mc{A}=(\LL/\KK,\sigma,\gamma)$ as
\begin{equation*}
    \mc{A}=(\LL/\KK,\sigma,\gamma)= \left\{x_0 + x_1\msf{e} +\ldots + x_{n-1}\msf{e}^{n-1}|x_0,\ldots,x_{n-1}\in\LL \right\},
\end{equation*}
where $\msf{e}^n=\gamma\in\KK$ and $\lambda\msf{e}=\msf{e}\sigma(\lambda)$ for $\lambda\in\LL$. $\mc{A}$ is said to be a division algebra if every non-zero element of $\mc{A}$ is invertible. A {\em cyclic division algebra} is a cyclic algebra that is at the same time a division algebra. In the space-time coding literature (see \cite{oggier07} and reference therein), a cyclic division algebra is usually constructed from a cyclic algebra $\mc{A}=(\LL/\KK,\sigma,\gamma)$ with carefully chosen $\gamma$ such that none of $\gamma,\gamma^2,\ldots,\gamma^{n-1}$ are norms of some element of $\LL$.

%Let $\mc{A}$ be the cyclic division algebra as follows,
%\begin{equation}
%    \mc{A}=(\LL/\KK,\sigma,\gamma)= \left\{x_0 + x_1\msf{e} +\ldots + x_{n-1}\msf{e}^{n-1}|x_0,\ldots,x_{n-1}\in\KK\right\},
%\end{equation}
%where $\msf{e}^n=\gamma$ and $\lambda\msf{e}=\msf{e}\lambda$ for $\lambda\in\KK$.

Consider $n_t=n_r=T=n$, an $n\times n$ STBC carved from $\mc{A}$ corresponds to a finite subset of
\begin{align}\label{eqn:SP_code_Abar}
    \bar{\mc{A}}_{\mfk{I}} = \left\{  x_0 + x_1\msf{e} +\ldots + x_{n-1}\msf{e}^{n-1} | x_0,x_1, \ldots, x_{n-1}\in\mfk{I} \right\},
\end{align}
where $\mfk{I}$ is an ideal in $\OL$. More specifically, an $n\times n$ STBC thus constructed can be obtained by putting $\bar{\mc{A}}_{\mfk{I}}$ into the matrix form given by
\begin{equation}\label{eqn:SP_code_matrix}
    \mc{C}_{\mfk{I}} = \left\{ \left.
                            \begin{pmatrix}
                              x_0                     & x_1                     & \ldots & x_{n-1} \\
                              \gamma\sigma(x_{n-1})   & \sigma(x_0)             &        & \sigma(x_{n-2}) \\
                              \vdots                  &                         & \ddots & \vdots \\
                              \gamma\sigma^{n-1}(x_1) & \gamma\sigma^{n-1}(x_2) & \ldots & \sigma^{n-1}(x_0) \\
                            \end{pmatrix}
      \right| x_0, \ldots, x_{n-1}\in\mfk{I} \right\}.
\end{equation}
A layer $\ell\in\{0,\ldots,n-1\}$ of the codeword in $\mc{C}_{\mfk{I}}$ is the collection of the entries in positions $(m,(\ell+m)\mod (n))$ for $m\in\{1,\ldots,n\}$. We note that each layer $\ell\in\{0,\ldots,n-1\}$ corresponds to the same $x_{\ell}\in \mfk{I}$. Here, we use the subscript $\mfk{I}$ in $\bar{\mc{A}}_{\mfk{I}}$ and $\mc{C}_{\mfk{I}}$ to emphasize that the elements $x_{\ell}$ for all $\ell$ are restricted to the ideal $\mfk{I}$. For transmission with finite input power constraint, one carves a subset from (a possibly shifted and scaled version of) $\mc{C}_{\mfk{I}}$ to form the codebook. From this point onward, we restrict the discussion to $\KK=\Q(\msf{i})$ or $\Q(\omega)$, which corresponds to the case where each $x_{\ell}$ is a linear combination of $n$ QAM or HEX constellation symbols. One observes that each codeword $\mbf{X}\in\mc{C}_{\mfk{I}}$ conveys $n$ symbols of $\LL$, where each symbol $x_{\ell}$ is a linear combination of $n$ QAM or HEX symbols. Therefore, the STBC thus constructed is \textit{full-rate}. i.e., it uses an $n\times n$ matrix to transmit $n^2$ symbols. Another consequence of having each $x_{\ell}$ being a linear combination of $n$ QAM or HEX symbols is that the code may not be energy-efficient as compared to sending QAM or HEX symbols directly. This drawback can often be overcome by choosing a suitable ideal $\mfk{I}$ such that $\mc{C}_{\mfk{I}}$ becomes a scaled and rotated version of $\Zi^n$ or $\Zw^n$.

The determinant of the codeword $\mbf{X}\in\mc{C}_{\mfk{I}}$ corresponding to $x\in\mc{A}$ is called the {\em reduced norm} of $x$. What is important about having the structure of cyclic division algebra is that when $\gamma\in\Ok$ not the norm of an element in $\LL$, it guarantees that the code is \textit{fully diverse} and has \textit{non-vanishing determinant} (NVD). This is evident from \cite[Corollary1 and Corollary 2]{Oggier06perfect}, which states that the reduced norm of $x\in \bar{\mc{A}}_{\OL}$ belongs to $\Ok$ and thus $\delta(\mc{C}_{\OL})=1$. Now, since $\mfk{I}\subseteq \OL$, one has that $\delta(\mc{C}_{\mfk{I}})\geq 1$. In fact, one can obtain better bounds on $\delta(\mc{C}_{\mfk{I}})$ as follows.
\begin{lemma}[{\cite[Corollary 3 and Corollary 4]{Oggier06perfect}}]\label{lma:bounds_delta_I}
    Let $\mc{C}_{\mfk{I}}$ be a STBC built over the cyclic division algebra $\mc{A}=(\LL/\KK,\sigma,\gamma)$ as in \eqref{eqn:SP_code_matrix}, where $\gamma\in\Ok$ not the norm of an element in $\LL$. Then,
\begin{equation*}
    N(\mfk{I}) \leq \delta(\mc{C}_{\mfk{I}}) \leq \min_{x\in\mfk{I}} N_{\LL/\Q}(x).
\end{equation*}
\end{lemma}

We end this section by providing the definition of a {\it perfect} STBC as follows.
\begin{define}
    A $n\times n$ STBC is called a perfect STBC if {\em i)} it is full-rate; {\em ii)} it is fully diverse and has NVD property; {\em iii)} the energy used to send the coded symbol on each layer is equal to that for sending the uncoded symbol themselves; and {\em iv)} all the coded symbols have the same average energy.
%\begin{itemize}
%  \item it is full-rate;
%  \item it is fully diverse and has NVD property;
%  \item the energy used to send the coded symbol on each layer is similar to that for sending the uncoded symbol themselves;
%  \item all the coded symbols have the same average energy.
%\end{itemize}
\end{define}
%\subsection{Golden-coded index coding in \cite{huang17GCIC}}

%=========================================================================
\section{Proposed Layered Space-Time Index Coding}\label{sec:LSTIC}
In this section, we propose the LSTIC and show that for any index set, it can provide SNR gain that is proportional to the information contained in the side information. In the proposed scheme, instead of directly tackling $\bar{\mc{A}}_{\OL}$ as done in \cite{huang17GCIC}, we recognize the layered structure of STBC reviewed in Section~\ref{sec:CDA_review} and perform partition layer by layer. More specifically, we split each message $w_k$, $k\in\{1,\ldots,K\}$, into $n$ sub-messages, namely $w_{k,\ell}$ for $\ell\in\{0,\ldots,n-1\}$, and encode $w_{1,\ell},\ldots,w_{K,\ell}$ into $x_{\ell}$ the layer $\ell$. The main advantage of this approach is that now each layer's signal is in $\OL$ and thereby one can apply CRT for partitioning. In what follows, we focus solely on cyclic division algebras with $\gamma\in\Ok$, such that none of $\gamma,\gamma^2,\ldots,\gamma^{n-1}$ are norms of element in $\LL$. We split the discussion into two parts depending on whether $\mfk{I}$ is principal or not. The first case includes constructions from $2 \times 2$, $3\times 3$, and $4\times 4$ perfect STBC while the second case encompasses constructions from the $6\times 6$ perfect STBC. The similar approach can also be applied to Alamouti code for constructing Layered Alamouti-coded index coding, which will be discussed in Section~\ref{sec:AlaIC}.

\begin{remark}
We emphasize that the approach that we propose in the following in fact applies to any cyclic division algebra with the non-norm element $\gamma$ with $\KK=\Q(\msf{i})$ or $\Q(\omega)$. For instance, the STBC design with non-norm element $\gamma \in\KK$ in \cite{elia07} can also be used as the base STBC of our LSTIC. The main reason that we particularly focus on $\gamma\in\Ok$ is so that we can rely on Lemma~\ref{lma:bounds_delta_I} to prove a lower bound on the side information gain. Apart from this, the proposed method does not require $\gamma\in\Ok$.
\end{remark}

\subsection{LSTIC with principal $\mfk{I}$}\label{subsec:LSTIC_principal}
Without loss of generality, we assume that $\mfk{I}$ is generated by some $\alpha\in\OL$, i.e., $\mfk{I}=\alpha\OL$. Then, \eqref{eqn:SP_code_Abar} becomes
\begin{align*}
    \bar{\mc{A}}_{\mfk{I}} &= \left\{ x_0 + x_1\msf{e} +\ldots + x_{n-1}\msf{e}^{n-1} | x_0,x_1, \ldots, x_{n-1}\in \alpha\OL \right\}, \nn \\
                &=\left\{ \alpha x_0 + \alpha x_1\msf{e} +\ldots + \alpha x_{n-1}\msf{e}^{n-1} | x_0,x_1, \ldots, x_{n-1}\in \OL \right\},
\end{align*}
and \eqref{eqn:SP_code_matrix} can be rewritten as
\begin{equation}\label{eqn:SP_code_matrix_principal}
      \left\{ \left.
        D(\alpha)\cdot
        \begin{pmatrix}
                              x_0                     & x_1                     & \ldots & x_{n-1} \\
                              \gamma\sigma(x_{n-1})   & \sigma(x_0)             &        & \sigma(x_{n-2}) \\
                              \vdots                  &                         & \ddots & \vdots \\
                              \gamma\sigma^{n-1}(x_1) & \gamma\sigma^{n-1}(x_2) & \ldots & \sigma^{n-1}(x_0) \\
                            \end{pmatrix}
      \right| x_0,\ldots, x_{n-1}\in\OL \right\},
\end{equation}
where
\begin{equation*}
  D(\alpha) \defeq \begin{pmatrix}
          \alpha & 0  & \ldots & 0 \\
          0 & \sigma(\alpha) &  \ldots & 0 \\
          \vdots  & \vdots &  \ddots & \vdots \\
          0 & 0 & \ldots  & \sigma^{n-1}(\alpha) \\
        \end{pmatrix}
\end{equation*}
We emphasize here that, as mentioned in Section~\ref{sec:CDA_review}, the codebook that we actually use should be a scaled version of the above codebook to satisfy the power constraint. However, in our analysis, what we really care is the {\it ratio} between the minimum determinants of the codebooks with and without side information, where the scaling does not make any difference. Therefore, throughout the paper, when analyzing the proposed scheme, we ignore the scaling factor for the sake of brevity. On the other hand, in our simulations, we do take the scaling into account and normalize the codebook to make the parameters reflect the actual $\SNR$.

We can now use the technique in \cite{huang15_lic} to partition $\OL$. Let $\mfk{q}_1,\ldots,\mfk{q}_K$ be $K$ ideals in $\OL$ that are relatively prime and have $N(\mfk{q}_k)=q_k$, $k\in\{1,\ldots,K\}$. Note that $\mfk{q}_k$s are not necessarily prime ideals and $q_k$s are not necessarily prime. We have $\mfk{q}_1\cap\ldots\cap\mfk{q}_K =\mfk{q_1}\cdot \ldots\cdot \mfk{q}_K\defeq \mfk{q}$. From CRT, we have
\begin{equation*}
    \OL/\mfk{q} \cong \OL/\mfk{q}_1\times \ldots \times \OL/\mfk{q}_K \cong \mbb{B}_{q_1}\times\ldots\times\mbb{B}_{q_K},
\end{equation*}
where $\mbb{B}_{q_k}=\OL/\mfk{q}_k$ is a commutative ring\footnote{Depending on the ideal $\mfk{q}_k$, this ring could be a finite field, a product of finite fields, a product of finite rings and finite fields, or others. But it is always commutative since a quotient ring of a commutative ring is always commutative. Throughout the paper, we do not use the ring property of the messages and therefore, we do not emphasize which type of ring it is.} with size $q_k$. Let $\mc{M}$ be an isomorphism that maps $\mbb{B}_{q_1}\times\ldots\times\mbb{B}_{q_K}$ to a complete set of coset leaders of $\OL/\mfk{q}$ having minimum energy.

Now, for $k\in\{1,\ldots,K\}$, let $w_k\in\mbb{B}_{q_k}^n$ which can be represented as $w_k=(w_{k,0},\ldots,w_{k,n-1})$ where each $w_{k,\ell}\in\mbb{B}_{q_k}$. The encoder collects $w_{1,\ell},\ldots,w_{K,\ell}$ to form the signal of the layer $\ell\in\{0,\ldots,n-1\}$ as
\begin{equation*}
    x_{\ell} = \mc{M}(w_{1,\ell},\ldots,w_{K,\ell})\in \OL/\mfk{q},\quad \ell\in\{0,\ldots,n-1\}.
\end{equation*}
The overall codebook corresponds to
\begin{equation*}
    \bar{\mc{A}}=\left\{ \alpha x_0 + \alpha x_1\msf{e} +\ldots + \alpha x_{n-1}\msf{e}^{n-1} | x_0, \ldots, x_{n-1}\in \OL/\mfk{q} \right\},
\end{equation*}
a subset of $\bar{\mc{A}}_{\mfk{I}}$ and has the matrix form as that in \eqref{eqn:SP_code_matrix_principal} with $x_0,\ldots,x_{n-1}\in\OL/\mfk{q}$.

For the proposed LSTIC within this class, we can show the following theorem.
\begin{theorem}\label{thm:SI_gain_principal}
    For any $\mc{S}\subset\{1,\ldots,K\}$, the proposed LSTIC with principal $\mfk{I}$ provides a side information gain at least $6$ dB/bits per real symbol, i.e., $\Gamma(\mc{C},\mc{S})\geq 6$ dB/bits per real symbol. Moreover, if all $\mfk{q}_k$, $k\in\{1,\ldots,K\}$, are principal, then $\Gamma(\mc{C},\mc{S}) = 6$ dB/bits per real symbol.
\end{theorem}
\begin{proof}
    We first note that in the proposed scheme, each message is spread onto $n$ layers of signals, which are then mapped to a $n\times n$ complex codeword matrix. i.e., $2n^2$ real symbols. Therefore, the rate of the message $w_k$ is given by
    \begin{equation}\label{eqn:R_k_principal}
        R_k = \frac{1}{2n^2}\log_2(q_k^n) = \frac{1}{2n}\log_2(q_k),\quad \text{bits per real symbol}.
    \end{equation}

    Consider a generic receiver with index set $\mc{S}$, let the messages be $w_s=v_s$ for $s\in\mc{S}$. This means that $w_{s,\ell} = v_{s,\ell}$ for all $\ell\in\{0,\ldots,n-1\}$ are known at the receiver. Let us first take $\mc{S}=\{s\}$ for example. The $\ell$th layer's signal can then be rewritten as
    \begin{align*}
      x_{\ell}^{\{s\}} &= \mc{M}(w_{1,\ell},\ldots,w_{s-1,\ell},v_{s,\ell},w_{s+1,\ell},\ldots,w_{K,\ell}) \nn \\
       &\overset{(a)}{=} \mc{M}(0,\ldots,0,v_s,0,\ldots,0) + \mc{M}(w_{1,\ell},\ldots,w_{s-1,\ell}, 0 ,w_{s+1,\ell},\ldots,w_{K,\ell}) + \zeta_{\ell}^{\{s\}} \nn \\
       &= \xi_{\ell}^{\{s\}} + \tilde{x}_{\ell}^{\{s\}},
    \end{align*}
    where $\zeta_{\ell}^{\{s\}}\in\mfk{q}$, $\tilde{x}_{\ell}^{\{s\}}\defeq \mc{M}(w_{1,\ell},\ldots,w_{s-1,\ell}, 0 ,w_{s+1,\ell},\ldots,w_{K,\ell})+\zeta_{\ell}^{\{s\}}$, and $\xi_{\ell}^{\{s\}}\defeq \mc{M}(0,\ldots,0,v_s,0,\ldots,0)$ is known at the receiver. The equality (a) above holds because $\mc{M}$ is an isomorphism. From CRT, we have
    \begin{align*}
      (x_{\ell}^{\{s\}} - \xi_{\ell}^{\{s\}})\mod \mfk{q}_s &= 0,
    \end{align*}
    which implies that $x_{\ell}^{\{s\}}$ belongs to a shifted version of $\mfk{q}_s$. For the general $\mc{S}$, we can similarly show that
    \begin{align}\label{eqn:x_given_SI_principal}
      x_{\ell}^{\mc{S}} &= \mc{M}(d_{1,\ell},\ldots,d_{K,\ell}) + \mc{M}(u_{1,\ell},\ldots,u_{K,\ell}) + \zeta_{\ell}^{\mc{S}} \nn \\
      &=\xi_{\ell}^{\mc{S}} + \tilde{x}_{\ell}^{\mc{S}} ,
    \end{align}
    where $\zeta_{\ell}^{\mc{S}}\in\mfk{q}$, $\tilde{x}_{\ell}^{\mc{S}}=\mc{M}(u_{1,,\ell},\ldots,u_{K,\ell})+ \zeta_{\ell}^{\mc{S}}$, and $\xi_{\ell}^{\mc{S}} \defeq \mc{M}(d_{1,\ell},\ldots,d_{K,\ell})$ with
    \begin{equation}\label{eqn:dk_principal}
        d_{k,\ell}=\left\{
                    \begin{array}{ll}
                     v_{k,\ell} , & \hbox{$k\in\mc{S}$;} \\
                      0 , & \hbox{$k\in\mc{S}^c$,}
                    \end{array}
                  \right.
    \end{equation}
    and
    \begin{equation}\label{eqn:uk_principal}
        u_{k,\ell}= \left\{
                      \begin{array}{ll}
                        0, & \hbox{$k\in\mc{S}$;} \\
                        w_{k,\ell}, & \hbox{$k\in\mc{S}^c$.}
                      \end{array}
                    \right.
    \end{equation}
    Note that $\xi_{\ell}^{\mc{S}}$ is known at the receiver. We now have
    \begin{equation*}
      \left(x_{\ell}^{\mc{S}}- \xi_{\ell}^{\mc{S}}\right)\mod \mfk{q}_s = 0,\quad\text{for all $s\in\mc{S}$},
    \end{equation*}
    which shows that $x_{\ell}^{\mc{S}}$ belongs to a shifted version of $\cap_{s\in\mc{S}}\mfk{q}_s = \Pi_{s\in\mc{S}}\mfk{q}_s$. Therefore, after revealing $w_{\mc{S}}$, the code $\mc{C}_{\mc{S}}$ corresponds to
    \begin{align*}
        &\left\{ \alpha(\xi_0^{\mc{S}} + \ldots + \xi_{n-1}^{\mc{S}}\msf{e}^{n-1}) + \alpha ( \tilde{x}_0^{\mc{S}} +\ldots + \tilde{x}_{n-1}^{\mc{S}} \msf{e}^{n-1}) | \tilde{x}_0^{\mc{S}},\ldots, \tilde{x}_{n-1}^{\mc{S}}\in \Pi_{s\in\mc{S}}\mfk{q}_s \right\},
    \end{align*}
    Hence, thanks to that $\sigma$ is a homomorphism, each codeword $\mathbf{X}\in\mc{C}_{\mc{S}}$ has the matrix form given by
    \begin{align*}
        \mbf{X} &=  \mbf{V}^{\mc{S}} + \tilde{\mbf{X}}^{\mc{S}},
    \end{align*}
    where
    \begin{equation*}
        \mbf{V}^{\mc{S}} = D(\alpha)\cdot \begin{pmatrix}
                              \xi_0^{\mc{S}}                     & \xi_1^{\mc{S}}                     & \ldots & \xi_{n-1}^{\mc{S}} \\
                              \gamma\sigma(\xi_{n-1}^{\mc{S}})   & \sigma(\xi_0^{\mc{S}})             &        & \sigma(\xi_{n-2}^{\mc{S}}) \\
                              \vdots                  &                         & \ddots & \vdots \\
                              \gamma\sigma^{n-1}(\xi_1^{\mc{S}}) & \gamma\sigma^{n-1}(\xi_2^{\mc{S}}) & \ldots & \sigma^{n-1}(\xi_0^{\mc{S}}) \\
                            \end{pmatrix},
    \end{equation*}
    and
    \begin{equation*}
        \tilde{\mbf{X}}^{\mc{S}} =D(\alpha)\cdot \begin{pmatrix}
                              \tilde{x}_0^{\mc{S}}                     & \tilde{x}_1^{\mc{S}}                     & \ldots & \tilde{x}_{n-1}^{\mc{S}} \\
                              \gamma\sigma(\tilde{x}_{n-1}^{\mc{S}})   & \sigma(\tilde{x}_0^{\mc{S}})             &        & \sigma(\tilde{x}_{n-2}^{\mc{S}}) \\
                              \vdots                  &                         & \ddots & \vdots \\
                              \gamma\sigma^{n-1}(\tilde{x}_1^{\mc{S}}) & \gamma\sigma^{n-1}(\tilde{x}_2^{\mc{S}}) & \ldots & \sigma^{n-1}(\tilde{x}_0^{\mc{S}}) \\
                            \end{pmatrix}.
    \end{equation*}
    Note that the second part of $\tilde{\mbf{X}}^{\mc{S}}$ is a codeword of the code
    \begin{equation*}
    \mc{C}_{\Pi_{s\in\mc{S}}\mfk{q}_s} = \left\{ \left.
                            \begin{pmatrix}
                              x_0                     & x_1                     & \ldots & x_{n-1} \\
                              \gamma\sigma(x_{n-1})   & \sigma(x_0)             &        & \sigma(x_{n-2}) \\
                              \vdots                  &                         & \ddots & \vdots \\
                              \gamma\sigma^{n-1}(x_1) & \gamma\sigma^{n-1}(x_2) & \ldots & \sigma^{n-1}(x_0) \\
                            \end{pmatrix}
      \right| x_0, \ldots, x_{n-1}\in  \Pi_{s\in\mc{S}}\mfk{q}_s \right\},
    \end{equation*}
    whose minimum determinant can be bounded by Lemma~\ref{lma:bounds_delta_I} as follows,
    \begin{equation}\label{eqn:det_C_prod}
        \delta(\mc{C}_{\Pi_{s\in\mc{S}}\mfk{q}_s})\geq N({\Pi_{s\in\mc{S}}\mfk{q}_s}).
        %N({\Pi_{s\in\mc{S}}\mfk{p}_s})\leq \delta(\mc{C}_{\Pi_{s\in\mc{S}}\mfk{p}_s})\leq \min_{x\in{\Pi_{s\in\mc{S}}\mfk{p}_s}} N_{\LL/\Q}(x).
    \end{equation}
    The receiver can now subtract the known $\mbf{V}^{\mc{S}}$ and compute the minimum determinant as
    \begin{align*}
        \delta(\mc{C}_{\mc{S}}) &= |\det(D(\alpha))|^2 \delta(\mc{C}_{\Pi_{s\in\mc{S}}\mfk{q}_s}) \nn \\
        &= |N_{\LL/\KK}(\alpha)|^2\delta(\mc{C}_{\Pi_{s\in\mc{S}}\mfk{q}_s})\nn \\
        &\overset{(a)}{=}N(\alpha)\delta(\mc{C}_{\Pi_{s\in\mc{S}}\mfk{q}_s}),
    \end{align*}
    where (a) follows from the fact that $\KK=\Q(\msf{i})$ or $\Q(\omega)$ is a quadratic extension. Plugging \eqref{eqn:det_C_prod} into the above equation results in
    \begin{align}\label{eqn:delta_Cs_principal}
      \delta(\mc{C}_{\mc{S}}) &\geq N(\alpha) N(\Pi_{s\in\mc{S}}\mfk{q}_s) \nn \\
      &= N(\alpha) \Pi_{s\in\mc{S}} N(\mfk{q}_s)=N(\alpha) \Pi_{s\in\mc{S}} q_s,
    \end{align}
    where the last equality follows from the fact that the ideal norm is multiplicative. Moreover, without revealing any side information, the overall codebook would have
    \begin{equation}\label{eqn:delta_0_principal}
        \delta(\mc{C}) = N(\alpha)N(1) = N(\alpha).
    \end{equation}
    Combining \eqref{eqn:R_k_principal}, \eqref{eqn:delta_Cs_principal}, and \eqref{eqn:delta_0_principal} results in
    \begin{align*}
      \Gamma(\mc{C},\mc{S}) &\geq \frac{10\log_{10}(\Pi_{s\in\mc{S}}q_s)}{n \frac{1}{2n}\sum_{s\in\mc{S}} \log_2 (q_s)} \nn \\
       &=\frac{\sum_{s\in\mc{S}}20\log_{10}(q_s)}{\sum_{s\in\mc{S}} \log_2 (q_s)} = 6~\text{dB/bits per real symbol}.
    \end{align*}
    To prove the second statement, we note that if the ideal $\Pi_{s\in\mc{S}}\mfk{q}_s$ is principal, then we can indeed find an element in the ideal such that the inequality in \eqref{eqn:delta_Cs_principal} holds with equality. Hence, if $\mfk{q}_1,\ldots,\mfk{q}_K$ are all principal, $\Gamma(\mc{C},\mc{S})=6$ dB for every $\mc{S}$.
\end{proof}

\subsection{LSTIC with non-principal $\mfk{I}$}\label{subsec:LSTIC_nonprincipal}
We now construct LSTIC from a STBC based on a cyclic division algebra $\mc{A}=(\LL/\KK,\sigma,\gamma)$ and a non-principal ideal $\mfk{I}$ in $\OL$ as described in \eqref{eqn:SP_code_Abar}. Let $\mfk{q}_1,\ldots,\mfk{q}_K$ be $K$ ideals in $\OL$ that are relatively prime and have norm $N(\mfk{q}_k)=q_k$, $k\in\{1,\ldots,K\}$. We again let $\mfk{q}_1\cdot\ldots\cdot\mfk{q}_K=\mfk{q}$. We further assume that each $\mfk{q}_k$ and $\mfk{I}$ are relatively prime, which also implies that $\mfk{q}$ and $\mfk{I}$ are relatively prime. From the second isomorphism theorem \cite{Hungerford74} and CRT, we have
\begin{align*}%\label{eqn:2nd_iso}
    \mfk{I}/\mfk{I}\mfk{q} &\overset{(a)}{=} \mfk{I}/\mfk{I}\cap\mfk{q} \overset{(b)}{\cong} (\mfk{I}+\mfk{q})/\mfk{q} \nn \\
    &\overset{(c)}{=} \OL/\mfk{q}\overset{(d)}{\cong} \OL/\mfk{q}_1\times\ldots\times \OL /\mfk{q}_K \nn \\
    &\cong \mbb{B}_{q_1}\times\ldots\times\mbb{B}_{q_K},
\end{align*}
where both (a) and (c) are due to the fact that $\mfk{q}$ and $\mfk{I}$ are relatively prime, (b) follows from the second isomorphism theorem, and (d) follows from CRT. We use $\mbb{B}_{q_k}$ to denote the quotient ring that is isomorphic to $\OL/\mfk{q}_k$ which has size $q_k$. Let $\mc{M}$ be an isomorphism that maps elements in $\mbb{B}_{q_1}\times\ldots\times\mbb{B}_{q_K}$ to a complete set of coset leaders of $\mfk{I}/\mfk{I}\mfk{q}$.

For $k\in\{1,\ldots,K\}$, we again enforce $w_k=(w_{k,0},\ldots,w_{k,n-1})\in\mbb{B}_{q_k}^n$ where each $\ell\in\{0,\ldots,n-1\}$. The sub-messages $w_{1,\ell},\ldots,w_{K,\ell}$ are collected and encoded into $x_{\ell}$ the signal of the $\ell\in\{0,\ldots,n-1\}$ layer as
\begin{equation*}
    x_{\ell} = \mc{M}(w_{1,\ell},\ldots,w_{K,\ell})\in\mfk{I}/\mfk{I}\mfk{q}, \quad \ell\in\{0,\ldots,n-1\}.
\end{equation*}
The overall codebook now corresponds to $\{x_0+x_1\msf{e}+\ldots+x_{n-1}\msf{e}^{n-1}|x_0,\ldots,x_{n-1}\in\mfk{I}/\mfk{I}\mfk{q}\}$ a subset of $\bar{\mc{A}}_{\mfk{I}}$ and has the matrix form as that in \eqref{eqn:SP_code_matrix} with $x_0,\ldots,x_{n-1}\in\mfk{I}/\mfk{I}\mfk{q}$.

For the proposed LSTIC within this class, we can show the following theorem.
\begin{theorem}\label{thm:SI_gain_nonprincipal}
    For any $\mc{S}\subset\{1,\ldots,K\}$, the side information gain achieved by the proposed LSTIC with non-principal ideal $\mfk{I}$ is lower bounded as
    \begin{equation*}
        \Gamma(\mc{C},\mc{S}) \geq 6 + \gamma_{\mfk{I}}~\text{dB/bits per real symbol},
    \end{equation*}
    where
    \begin{equation}\label{eqn:gamma_I}
        \gamma_{\mfk{I}} = 20\log_{10}\left(\frac{N(\mfk{I})}{\min_{x\in\mfk{I}} N_{\LL/\mbb{Q}}(x)} \right),
    \end{equation}
    is negative and is only a function of $\mfk{I}$ and is independent of $\mc{S}$.
\end{theorem}
\begin{proof}
    We again note that the rate of the message $w_k$ is given by
    \begin{equation}\label{eqn:R_k_nonprincipal}
        R_k = \frac{1}{2n^2}\log_2(q_k^n) = \frac{1}{2n}\log_2(q_k),\quad \text{bits per real symbol}.
    \end{equation}

    We consider a generic receiver having index set $\mc{S}$. Suppose the messages $w_s=v_s$ for $s\in{\mc{S}}$ are known, which means that $w_{s,\ell}=v_{s,\ell}$ for all $\ell\in\{0,\ldots,n-1\}$ are known at the receiver. Similar to \eqref{eqn:x_given_SI_principal}, we have
    \begin{align*}
      x_{\ell}^{\mc{S}} &= \mc{M}(d_{1,\ell},\ldots,d_{K,\ell}) + \mc{M}(u_{1,\ell},\ldots,u_{K,\ell}) + \zeta_{\ell}^{\mc{S}} \nn \\
      &=\xi_{\ell}^{\mc{S}} + \tilde{x}_{\ell}^{\mc{S}},
    \end{align*}
    where $d_{k,\ell}$ and $u_{k,\ell}$ are defined in \eqref{eqn:dk_principal} and \eqref{eqn:uk_principal}, respectively, and $\zeta_{\ell}^{\mc{S}}\in\mfk{Iq}$. Therefore, we have
    \begin{equation*}
      \left(x_{\ell}^{\mc{S}}- \xi_{\ell}^{\mc{S}}\right)\mod \mfk{I}\mfk{q}_s = 0,\quad\text{for all $s\in\mc{S}$},
    \end{equation*}
    which means that $x_{\ell}^{\mc{S}}$ belongs to a shifted version of
    \begin{align*}
        \cap_{s\in\mc{S}}\mfk{I}\mfk{q}_s &\overset{(a)}{=} \cap_{s\in\mc{S}}\left(\mfk{I}\cap\mfk{q}_s\right) \nn \\
        &= \mfk{I}\cap\left(\cap_{s\in\mc{S}}\mfk{q}_s\right) \overset{(b)}{=} \mfk{I}\Pi_{s\in\mc{S}}\mfk{q}_s,
    \end{align*}
    where (a) follows from that $\mfk{I}$ and $\mfk{q}_s$ are relatively prime for each $s$ and (b) is due to the fact that $\mfk{q}_1,\ldots,\mfk{q}_K$ are relatively prime.

    After revealing $w_{\mc{S}}$, the code $\mc{C}_{\mc{S}}$ would correspond to
    \begin{align*}
        &\left\{ (\xi_0^{\mc{S}} + \ldots + \xi_{n-1}^{\mc{S}}\msf{e}^{n-1}) + ( \tilde{x}_0^{\mc{S}} +\ldots + \tilde{x}_{n-1}^{\mc{S}} \msf{e}^{n-1}) | \tilde{x}_0^{\mc{S}},\ldots, \tilde{x}_{n-1}^{\mc{S}}\in \mfk{I}\Pi_{s\in\mc{S}}\mfk{q}_s \right\},
    \end{align*}
    Therefore, each codeword $\mathbf{X}\in\mc{C}_{\mc{S}}$ has the matrix form given by
    \begin{align*}
        \mbf{X} &=  \mbf{V}^{\mc{S}} + \tilde{\mbf{X}}^{\mc{S}},
    \end{align*}
    where
    \begin{equation*}
        \mbf{V}^{\mc{S}} = \begin{pmatrix}
                              \xi_0^{\mc{S}}                     & \xi_1^{\mc{S}}                     & \ldots & \xi_{n-1}^{\mc{S}} \\
                              \gamma\sigma(\xi_{n-1}^{\mc{S}})   & \sigma(\xi_0^{\mc{S}})             &        & \sigma(\xi_{n-2}^{\mc{S}}) \\
                              \vdots                  &                         & \ddots & \vdots \\
                              \gamma\sigma^{n-1}(\xi_1^{\mc{S}}) & \gamma\sigma^{n-1}(\xi_2^{\mc{S}}) & \ldots & \sigma^{n-1}(\xi_0^{\mc{S}}) \\
                            \end{pmatrix},
    \end{equation*}
    and
    \begin{equation*}
        \tilde{\mbf{X}}^{\mc{S}} =\begin{pmatrix}
                              \tilde{x}_0^{\mc{S}}                     & \tilde{x}_1^{\mc{S}}                     & \ldots & \tilde{x}_{n-1}^{\mc{S}} \\
                              \gamma\sigma(\tilde{x}_{n-1}^{\mc{S}})   & \sigma(\tilde{x}_0^{\mc{S}})             &        & \sigma(\tilde{x}_{n-2}^{\mc{S}}) \\
                              \vdots                  &                         & \ddots & \vdots \\
                              \gamma\sigma^{n-1}(\tilde{x}_1^{\mc{S}}) & \gamma\sigma^{n-1}(\tilde{x}_2^{\mc{S}}) & \ldots & \sigma^{n-1}(\tilde{x}_0^{\mc{S}}) \\
                            \end{pmatrix}.
    \end{equation*}
    We can again note that $\tilde{\mbf{X}}^{\mc{S}}$ belongs to
    \begin{equation*}
        \mc{C}_{\mfk{I}\Pi_{s\in\mc{S}}\mfk{q}_s} =\left\{ \left.
                            \begin{pmatrix}
                              x_0                     & x_1                     & \ldots & x_{n-1} \\
                              \gamma\sigma(x_{n-1})   & \sigma(x_0)             &        & \sigma(x_{n-2}) \\
                              \vdots                  &                         & \ddots & \vdots \\
                              \gamma\sigma^{n-1}(x_1) & \gamma\sigma^{n-1}(x_2) & \ldots & \sigma^{n-1}(x_0) \\
                            \end{pmatrix}
      \right| x_0, \ldots, x_{n-1}\in  \mfk{I}\Pi_{s\in\mc{S}}\mfk{q}_s \right\},
    \end{equation*}
    whose minimum determinant can be bounded via Lemma~\ref{lma:bounds_delta_I} by
    \begin{equation*}
        \delta(\mc{C}_{\mfk{I}\Pi_{s\in\mc{S}}\mfk{q}_s}) \geq N(\mfk{I}\Pi_{s\in\mc{S}}\mfk{q}_s).
    \end{equation*}
    One can now remove the contribution of $\mbf{V}^{\mc{S}}$ from the received signal and bound the minimum determinant as
%    \begin{align}\label{eqn:det_X_Cs_nonprincipal}
%      \det(\tilde{\mbf{X}}^{\mc{S}}) &= \det \begin{pmatrix}
%                              \tilde{x}_0^{\mc{S}}                     & \tilde{x}_1^{\mc{S}}                     & \ldots & \tilde{x}_{n-1}^{\mc{S}} \\
%                              \gamma\sigma(\tilde{x}_{n-1}^{\mc{S}})   & \sigma(\tilde{x}_0^{\mc{S}})             &        & \sigma(\tilde{x}_{n-2}^{\mc{S}}) \\
%                              \vdots                  &                         & \ddots & \vdots \\
%                              \gamma\sigma^{n-1}(\tilde{x}_1^{\mc{S}}) & \gamma\sigma^{n-1}(\tilde{x}_2^{\mc{S}}) & \ldots & \sigma^{n-1}(\tilde{x}_0^{\mc{S}}) \\
%                            \end{pmatrix} \nn \\
%              &\overset{(a)}{\in} \mc{N}_{\LL/\KK}(\mfk{I}\Pi_{s\in\mc{S}}\mfk{p}_s),
%    \end{align}
%    where again $\mc{N}_{\LL/\KK}(\mfk{I}\Pi_{s\in\mc{S}}\mfk{p}_s)$ is the relative norm of the ideal $\mfk{I}\Pi_{s\in\mc{S}}\mfk{p}_s$ and (a) follows from \cite[Proposition 3]{Oggier06perfect}. Since $\KK=\mbb{Q}(i)$ or $\mbb{Q}(j)$, we have
%    \begin{align}\label{eqn:det_X_2_nonprincipal}
%      |\det(\tilde{\mbf{X}}^{\mc{S}})|^2 &\in \mc{N}_{\LL/\Q}(\mfk{I}\Pi_{s\in\mc{S}}\mfk{p}_s) = N(\mfk{I}\Pi_{s\in\mc{S}}\mfk{p}_s)\mbb{Z} = N(\mfk{I})N(\Pi_{s\in\mc{S}}\mfk{p}_s)\mbb{Z},
%    \end{align}
%    and hence
    \begin{align}\label{eqn:delta_Cs_nonprincipal}
      \delta(\mc{C}_{\mc{S}}) &\geq N(\mfk{I}) N(\Pi_{s\in\mc{S}}\mfk{q}_s) \nn \\
      &= N(\mfk{I}) \Pi_{s\in\mc{S}} N(\mfk{q}_s)= N(\mfk{I}) \Pi_{s\in\mc{S}} q_s.
    \end{align}
%    \begin{align}\label{eqn:det_X_2_nonprincipal}
%      |\det(\mbf{X})|^2 &\in \mc{N}_{\LL/\mbb{Q}}(\mfk{I}) = N(\mfk{I})\mbb{Z}.
%    \end{align}
    When no side information is available, we can again use Lemma~\ref{lma:bounds_delta_I} to bound the minimum determinant as
    \begin{align}\label{eqn:delta_0_nonprincipal}
      N(\mfk{I})\leq \delta(\mc{C}) &\leq  \min_{x\in\mfk{I}} N_{\LL/\mbb{Q}}(x).
    \end{align}
    Combining \eqref{eqn:R_k_nonprincipal}, \eqref{eqn:delta_Cs_nonprincipal}, and \eqref{eqn:delta_0_nonprincipal} results in
    \begin{align*}
      \Gamma(\mc{C},\mc{S}) &\geq \frac{10\log_{10}\left(  N(\Pi_{s\in\mc{S}}\mfk{q}_s)\frac{N(\mfk{I})}{\min_{x\in\mfk{I}} N_{\LL/\mbb{Q}}(x)}\right)}{n \frac{1}{2n}\sum_{s\in\mc{S}} \log_2 (p_s)} \nn \\
       &=\frac{\sum_{s\in\mc{S}}20\log_{10}(q_s)}{\sum_{s\in\mc{S}} \log_2 (q_s)} + \frac{20\log_{10}\left(\frac{N(\mfk{I}) }{min_{x\in\mfk{I}} N_{\LL/\mbb{Q}}(x)}\right)}{\sum_{s\in\mc{S}} \log_2 (p_s)} \nn \\
       &= 6 + \gamma_{\mfk{I},\mc{S}}~\text{dB/bits per real symbol}.
    \end{align*}
    Noting that $\gamma_{\mfk{I},\mc{S}}\leq 0$ from \eqref{eqn:delta_0_nonprincipal} and $\gamma_{\mfk{I},\mc{S}}\geq \gamma_{\mfk{I}}$ completes the proof.
   \end{proof}

\section{Layered Golden-Coded Index Coding}\label{sec:L_GIC}
In this section, we propose layered Golden-coded index coding, a family of LSTIC constructed from Golden code. To provide a concrete illustration of how the proposed scheme works, we will walk through this example in detail. Before proceeding, we note that the layered Golden-coded index coding proposed here is different, in essence, from the Golden-coded index coding in \cite{huang17GCIC}. Here, we partition the code layer by layer while in \cite{huang17GCIC} we directly tackle the Golden algebra. We would like to emphasize that neither of these two schemes subsumes the other as a special case; however, the approach taken in \cite{huang17GCIC} only works for some particular primes.

Let $\LL=\mbb{Q}(\msf{i},\sqrt{5})$ a quadratic extension of $\KK=\mbb{Q}(\msf{i})$ and consider the non-trivial $\mbb{Q}(\msf{i})$-automorphism $\sigma:\sqrt{5}\rightarrow -\sqrt{5}$. Also, let $\gamma=\msf{i}$. The Golden code is built from the Golden algebra given by
\begin{equation*}
    \mc{G}=(\mbb{Q}(\msf{i},\sqrt{5})/\mbb{Q}(i),\sigma,\msf{i}) = \left\{ x_0 + x_1 \msf{e}| x_0, x_1\in\mbb{Q}(i,\sqrt{5}) \right\},
\end{equation*}
where $\msf{e}^2=\msf{i}$ and $z\msf{e}=\msf{e}\sigma(z)$. The ring of integers of $\LL$ is $\OL=\mbb{Z}[\msf{i}][\theta]$ where $\theta=\frac{1+\sqrt{5}}{2}$. Let $\mfk{I}=\alpha\OL$ be the principal ideal generated by $\alpha=1+\msf{i}\bar{\theta}$ where $\bar{\theta}\defeq \sigma(\theta)$. The Golden code \cite{belfiore05GC} corresponds to
\begin{equation*}
    \mc{G}_{\mfk{I}} = \left\{x_0 + x_1 \msf{e} | x_0,x_1\in\alpha\OL \right\},
\end{equation*}
which can be put into the matrix form
\begin{align*}%\label{eqn:GC_CI}
    &\mc{C}_{\mfk{I}} = \left\{\left.\frac{1}{\sqrt{5}}\begin{pmatrix}
                      \alpha x_0  & \alpha x_1 \\
                      i\sigma(\alpha x_1)  & \sigma(\alpha x_0) \\
                    \end{pmatrix} \right| x_0,x_1\in\mbb{Z}[i][\theta]
    \right\} \nonumber \\
    &= \left\{\left.\frac{1}{\sqrt{5}}\begin{pmatrix}
                      \alpha (a+b\theta) & \alpha (c+d\theta) \\
                      i\sigma(\alpha) (c+d\bar{\theta}) & \sigma(\alpha)(a+b\bar{\theta}) \\
                    \end{pmatrix} \right|a,b,c,d\in\Zi
    \right\}.
\end{align*}

%Let $\mbb{K}=\mbb{Q}(i,\sqrt{5})$ and $\Ok = \mbb{Z}[i][\theta]$ be its ring of integers. We emphasize here that $\mbb{K}$ is a number field with degree 4 and hence has four $\mbb{Q}$-automorphism, namely $\sigma_1:(i,\sqrt{5})\rightarrow (i,\sqrt{5}), \sigma_2:(i,\sqrt{5})\rightarrow (i,-\sqrt{5})$ and  $\sigma_3=\bar{\sigma_1},\sigma_4=\bar{\sigma_2}$. The $\mbb{Q}(i)$-automorphism used in golden algebra is in fact $\sigma=\sigma_2$. Recall that a golden code corresponds to a finite subset of $\bar{\mc{A}}=\left\{x_0+x_1\msf{e}| x_0,x_1\in\Ok\right\}$. The idea is to split the signal into two parts, respectively mapped to $x_0$ and $x_1$, and partition $\Ok$ instead of directly partitioning $\bar{\mc{A}}$.

The proposed layered Golden-coded index coding can be categorized into the class in Section~\ref{subsec:LSTIC_principal}. Let $\mfk{q}_1,\mfk{q}_2,\ldots,\mfk{q}_K$ be prime ideals in $\OL$ that are relatively prime. Let $\mfk{q}_1\ldots\mfk{q}_K\defeq \mfk{q}$. Also, let $|\Ok/\mfk{q}_k|=N(\mfk{q}_k)\defeq q_k$ for $k\in\{1,\ldots,K\}$ where $q_k$s are not necessarily primes. From CRT, we have
\begin{align*}
    \Ok/\mfk{q} \cong \Ok/\mfk{q}_1 \times\ldots\times \Ok/\mfk{q}_K \cong \mbb{B}_{q_1} \times \ldots\times\mbb{B}_{q_K},
\end{align*}
where $\mbb{B}_{q_k}=\Ok/\mfk{q}_k$ is a commutative ring with size $q_k$. This guarantees the existence of $\mc{M}:\mbb{B}_{q_1} \times\ldots\times\mbb{B}_{q_K} \rightarrow \OL/\mfk{q}$ an isomorphism that maps the messages to a complete set of coset leaders of $\OL/\mfk{q}$ with minimum energy. In the proposed layered Golden-coded index coding scheme, we let $w_k\in\mbb{B}_{q_k}^2$ and split it into $w_{k,0},w_{k,1}\in\mbb{B}_{q_k}$.

The sub-messages $w_{1,\ell},\ldots,w_{K,\ell}$, for $\ell\in\{0,1\}$, are encoded onto $\OL/\mfk{q}$ via $\mc{M}$ to form
\begin{equation}\label{eqn:GCIC}
    x_{\ell} = \mc{M}(w_{1,\ell},\ldots,w_{K,\ell}) \in \OL/\mfk{q},\quad \ell\in\{0,1\}.
\end{equation}
The overall codebook becomes a Golden code
\begin{align}\label{eqn:GCIC_matrix}
    \mc{C} = \left\{\left.\frac{1}{\sqrt{5}}\begin{pmatrix}
                      \alpha x_0  & \alpha x_1 \\
                      i\sigma(\alpha x_1)  & \sigma(\alpha x_0) \\
                    \end{pmatrix} \right| x_0,x_1\in \OL/\mfk{q}
    \right\}.
\end{align}

From Theorem~\ref{thm:SI_gain_principal}, we obtain the following corollary. Note that the proof of this corollary is almost identical to that of Theorem~\ref{thm:SI_gain_principal}. However, as mentioned earlier, in order to provide a complete illustration, we still present the proof.
\begin{corollary}
    For any $\mc{S}\subset\{1,\ldots,K\}$, the proposed layered Golden-coded index coding provides $\Gamma(\mc{C},\mc{S}) = 6$ dB/bits per real symbol.
\end{corollary}

\begin{proof}
The rate of the message $w_k$ is given by
\begin{align}\label{eqn:R_k}
  R_{k} &= \frac{1}{8} \log_2(N(\mfk{q}_k)^2)\quad\text{bits per real symbol.}
\end{align}

Suppose some messages $w_{\mc{S}}\defeq\{w_k=v_k|k\in\mc{S}\}$ are known; this means that both $w_{\mc{S},\ell}\defeq\{w_{k,\ell}=v_{k,\ell}|k\in\mc{S}\}$ for $\ell=0$ and $\ell=1$ are known. Therefore, from Section~\ref{subsec:LSTIC_principal}, $x_{\ell}$, $\ell\in\{0,1\}$, belongs to a shifted version of $\Pi_{k\in\mc{S}}\mfk{q}_k$. Thus, after revealing $w_{\mc{S}}$, the code $\mc{C}_{\mc{S}}$ becomes a shifted version of
\begin{align*}
    &\left\{\left.\frac{1}{\sqrt{5}}\begin{pmatrix}
                      \alpha x_0  & \alpha x_1 \\
                      i\sigma(\alpha x_1)  & \sigma(\alpha x_0) \\
                    \end{pmatrix} \right| x_0,x_1\in \Pi_{k\in\mc{S}}\mfk{q}_k
    \right\}.
    %\nonumber \\
%    &=\left\{\left.\frac{1}{\sqrt{5}}\begin{pmatrix}
%                      \alpha \beta_0\Pi_{k\in\mc{S}}\phi_k  & \alpha \beta_1\Pi_{k\in\mc{S}}\phi_k \\
%                      i\sigma(\alpha \beta_1\Pi_{k\in\mc{S}}\phi_k)  & \sigma(\alpha \beta_0\Pi_{k\in\mc{S}}\phi_k) \\
%                    \end{pmatrix} \right| \beta_0,\beta_1\in \Ok
%    \right\}.
\end{align*}
For every codeword $\tilde{\mathbf{X}}^{\mc{S}}\in\mc{C}_{\mc{S}}$ corresponding to $x_0,x_1\in \Pi_{k\in\mc{S}}\mfk{q}_k$, the determinant is given by
\begin{align*}
  \det(\tilde{\mathbf{X}}^{\mc{S}}) &= \frac{1}{5}\det \begin{pmatrix}
                      \alpha x_0  & \alpha x_1 \\
                      i\sigma(\alpha x_1)  & \sigma(\alpha x_0) \\
                    \end{pmatrix}
   \nonumber \\
   %\label{eqn:det_X}
   &\overset{(a)}{=} \frac{1}{5}\det\begin{pmatrix}
                                              \alpha x_0 &  \alpha x_1\\
                                              i\sigma(\alpha)\sigma(x_1) & \sigma(\alpha)\sigma(x_0)) \\
                                            \end{pmatrix}
   \\
   &=\frac{1}{5}\det\begin{pmatrix}
                      \alpha & 0 \\
                      0 & \sigma(\alpha) \\
                    \end{pmatrix}
   \det\begin{pmatrix}
                                              x_0 &  x_1 \\
                                              i\sigma(x_1) & \sigma(x_0) \\
                                            \end{pmatrix}
   \nonumber \\
   &=\frac{1}{5}N_{\textrm{rd}}(\alpha)
   \det\begin{pmatrix}
                                              x_0 &  x_1 \\
                                              i\sigma(x_1) & \sigma(x_0) \\
                                            \end{pmatrix},
   %\nonumber \\
%   \label{eqn:det_X_2}
%   &\geq \frac{1}{5}N_{\textrm{rd}}(\alpha)
%   N(\Pi_{k\in\mc{S}}\mfk{I}_k),
\end{align*}
where (a) is due to that $\sigma$ is a homomorphism. Now, plugging $|N_{\textrm{rd}}(\alpha)|^2=5$ results in
\begin{align}
  \delta(\mc{C}_{\mc{S}}) &= \frac{1}{5} \left|\det\begin{pmatrix}
                                              x_0 &  x_1 \\
                                              i\sigma(x_1) & \sigma(x_0) \\
                                            \end{pmatrix}\right|^2 \nonumber \\
  &\overset{(a)}{=} \frac{1}{5} N(\Pi_{k\in\mc{S}}\mfk{q}_k) \overset{(b)}{=} \frac{1}{5}\Pi_{k\in\mc{S}}N(\mfk{q}_k), \label{eqn:GC_delta_S}
\end{align}
where (a) follows from \cite[Corollary 3]{Oggier06perfect} and the fact that $\OL=\Z[\msf{i}][\theta]$ is a principal ideal domain and (b) follows from the fact that algebraic norm is multiplicative. Now, combining what we have obtained in \eqref{eqn:R_k} and \eqref{eqn:GC_delta_S} and the fact that $\delta(\mc{C}) = 1/5$ result in
\begin{equation*}
    \Gamma(\mc{C},\mc{S}) = \frac{10\log_{10}(\Pi_{k\in\mc{S}}N(\mfk{q}_k))}{2\frac{1}{4}\sum_{k\in\mc{S}} \log_2 N(\mfk{q}_k)} = 6~\text{dB/bits per real symbol}.
\end{equation*}
\end{proof}
%\begin{tcolorbox}
%\begin{remark}
%    Since Corollary 3 in \cite{Oggier06perfect} works for other perfect codes with the base field being either $\mbb{Q}(i)$ or $\mbb{Q}(j)$, this approach is now generalizable to other perfect space codes in \cite{Oggier06perfect}. For an $n_t\times n_t$ perfect code in \cite{Oggier06perfect}, the rate in \eqref{eqn:R_k} becomes
%    \begin{equation}
%        R_k =\frac{1}{2n_t^2}\log_2(N(\mfk{p}_k)^{n_t})=\frac{1}{2n_t}\log_2(N(\mfk{p}_k)),
%    \end{equation}
%    and $\delta_{\mc{S}}/\geq \Pi_{k\in\mc{S}} N(\mfk{p}_k)\delta$ can be obtained by following \eqref{eqn:delta_S_2}. Combining these results in
%    \begin{equation}
%        \Gamma(\mc{C},\mc{S}) = \frac{10\log_{10}(\Pi_{k\in\mc{S}}N(\mfk{p}_k))}{n_t\frac{1}{2n_t}\sum_{k\in\mc{S}} \log_2 N(\mfk{p}_k)} = 6~\text{dB/bits per real dimension}.
%    \end{equation}
%\end{remark}
%\end{tcolorbox}

\subsection{Examples and Simulation Results}
In Table~\ref{tbl:GC}, we factorize each prime $p<100$ into prime ideals in $\OL$ via Magma \cite{magma}. Any pair of ideals in this table is relatively prime and thus $\mfk{q}_k$ can be chosen as product of some prime ideals that have not been selected for some $\mfk{q}_{k'}$, $k'\neq k$. In Table~\ref{tbl:GC}, we show ideals and their inertial degrees $f$. The ramification index of each prime ideal lying above $p\neq 2,5$ is 1 and is 2 for prime ideals lying above $2, 5$. This can be seen by observing that
\begin{equation*}
    d_{\LL} =  5^2\cdot 4^2,
\end{equation*}
which has prime factors 2 and 5. Moreover, since $\OL$ is a principal ideal domain, so every $p\OL$ can be factorized into principal prime ideals.
\begin{table}
\centering
\caption{Prime factorization of $p<100$ in $\mbb{Z}[i][\theta]$ where $\theta=\frac{1+\sqrt{5}}{2}$.}
\label{tbl:GC}
\begin{tabular}{|c|c|c|}
  \hline
  $p$ & $\mfk{p}$ & $f$\\
\hline
  2 & $(1+\msf{i})$ & 2 \\
  3 & $(\bar{\theta}-i\theta), (\bar{\theta}+\msf{i}\theta)$ & 2 \\
  5 & $(1+\msf{i}\bar{\theta}), (1-\msf{i}\bar{\theta})$ & 1 \\
  7 & $((1+\theta)+\msf{i}(1+\bar{\theta})), ((1+\theta)-\msf{i}(1+\bar{\theta}))$ & 2  \\
  11 & $(3\msf{i}\theta-\msf{i}), (3\msf{i}\bar{\theta}-\msf{i})$ & 2 \\
  13 & $(2+3\msf{i}), (2-3\msf{i})$ & 2 \\
  17 & $(4+\msf{i}), (4-\msf{i})$ & 2 \\
  19 & $(4\msf{i}\theta-i), (4\msf{i}\bar{\theta}-\msf{i})$ & 2  \\
  23 & $((3\bar{\theta}-1)+\msf{i}(3\theta-1)), ((3\bar{\theta}-1)-\msf{i}(3\theta-1))$ & 2  \\
  29 & $(2\msf{i}+\theta),(2\msf{i}+\bar{\theta}), (\bar{\theta}-2\msf{i}), (\theta-2\msf{i})$ & 1 \\
  31 & $(2-5\bar{\theta}),(2-5\theta)$ & 2  \\
  37 & $(6+\msf{i}), (6-\msf{i})$ & 2  \\
  41 & $(\bar{\theta}+\msf{i}(2\theta-1)), (\bar{\theta}-\msf{i}(2\theta-1)), (\theta-\msf{i}(2\theta-1)), (\theta+\msf{i}(2\theta-1))$  & 1  \\
  43 & $((4+\theta) + \msf{i}(4+\bar{\theta})) , ((4+\bar{\theta}) + \msf{i}(4+\theta))$ & 2 \\
  47 & $((2+3\theta)-\msf{i}(2+3\bar{\theta})),((2+3\bar{\theta)}-\msf{i}(2+3\theta))$  & 2 \\
  53 & $(7+2\msf{i}), (7-2\msf{i})$ & 2 \\
  59 & $(7\theta-2), (7\bar{\theta}-2)$  & 2 \\
  61 & $((2\bar{\theta}-1)+\msf{i}(\theta+1)), ((2\theta-1)+\msf{i}(\bar{\theta}+1)), ((2\theta-1)+\msf{i}(\theta+1)), ((2\bar{\theta}-1)+\msf{i}(\bar{\theta}+1))$  & 1 \\
  67 & $((5\bar{\theta}-1)+\msf{i}(5\theta-1)), ((5\theta-1)+\msf{i}(5\bar{\theta}-1))$ & 2 \\
  71 & $(8+\theta), (8+\bar{\theta})$  & 2 \\
  73 & $(3+8\msf{i}) , (3-8\msf{i})$  & 2 \\
  79 & $(8\theta - 3) , (8\bar{\theta} - 3)$ & 2 \\
  83 & $((4+3\theta)+\msf{i}(4+3\bar{\theta})), ((4+3\bar{\theta})+\msf{i}(4+3\theta))$ & 2 \\
  89 & $(2\bar{\theta}-\msf{i}(\theta+1)), (2\theta-\msf{i}(\bar{\theta}+1)), (2\bar{\theta}+\msf{i}(\theta+1)), (2\bar{\theta}+\msf{i}(\bar{\theta}+1))$ & 1 \\
  97 & $(9+4\msf{i}) , (9-4\msf{i})$ & 2 \\
  \hline
\end{tabular}
\end{table}

Simulation results for the proposed layered Golden-coded index coding are provided in Fig.~\ref{fig:CER_GCIC_q3_q5}. In this figure, three sets of simulations are performed. In the first one, we constructed the layered Golden-coded index coding with two principal ideals generated by $\beta_1=(\bar{\theta}-i\theta)$ and $\beta_2=(\bar{\theta}+\msf{i}\theta)$, respectively. From Table~\ref{tbl:GC}, we see that each of these ideals corresponds to $p=3$ and has inertial degree 2; thus, it has norm equal to $3^2=9$. Thus, each message $w_k\in\mbb{B}_9^2$, which is then split into sub-messages $w_{k,1},w_{k,2}\in\mbb{B}_9$. The sub-messages $w_{1,\ell}$ and $w_{2,\ell}$ are then encoded into $x_{\ell}$ via \eqref{eqn:GCIC}, which is then put into the matrix form in \eqref{eqn:GCIC_matrix}. Moreover, from Table~\ref{tbl:GC}, we know that $3\OL = \beta_1\beta_2\OL$. Therefore, the overall codebook corresponds to \eqref{eqn:GCIC_matrix} with $x_0,x_1\in\OL/3\OL$. Simulation results in Fig.~\ref{fig:CER_GCIC_q3_q5} show that revealing either message to the receiver provides roughly 7.3 dB of SNR gain. This conforms with the analysis that when reveal either message, we expect to achieve SNR gain
\begin{equation*}
    \frac{1}{4} 10\log_{10}\left(\frac{118}{10}\right) + \frac{1}{2} 10\log_{10}\left( 9 \right) \approx 7.45~\text{dB},
\end{equation*}
where $118$ and $10$ inside the first logarithm are $N_{\mc{C}}$ and $N_{\mc{C}_{\mc{S}}}$, respectively and the $9$ inside the second logarithm is the ratio of $\delta(\mc{C}_{\mc{S}})$ and $\delta(\mc{C})$.

In the second set of simulations, the two principal ideals are replaced by those generated by $\beta_1=(1+\msf{i}\bar{\theta})^2$ and $\beta_2=(1-\msf{i}\bar{\theta})^2$, respectively. From Table~\ref{tbl:GC}, we see that $(1+\msf{i}\bar{\theta})$ and $(1-\msf{i}\bar{\theta})$ are both corresponding to $p=5$ with inertial degree 1; thus, $\beta_1\OL$ and $\beta_2\OL$ both have norm equal to $5^2=25$. Moreover, $5\OL = \beta_1\beta_2\OL$; thereby, the overall codebook corresponds to \eqref{eqn:GCIC_matrix} with $x_0,x_1\in\OL/5\OL$. Simulation results in Fig.~\ref{fig:CER_GCIC_q3_q5} show that revealing either message to the receiver provides roughly 10 dB of SNR gain. This again coincides with the analysis which says that by revealing one side information , we can expect an SNR gain of
\begin{equation*}
    \frac{1}{4} 10\log_{10}\left(\frac{656}{32}\right) + \frac{1}{2} 10\log_{10}\left( 25 \right) \approx 10.27~\text{dB},
\end{equation*}
where $656$ and $32$ inside the first logarithm are $N_{\mc{C}}$ and $N_{\mc{C}_{\mc{S}}}$, respectively and the $25$ inside the second logarithm is the ratio of $\delta(\mc{C}_{\mc{S}})$ and $\delta(\mc{C})$. In the last set of simulations, the two prime ideals corresponding to $p=7$ is considered. Simulation results show that a roughly 12.1 dB SNR gain can be obtained by revealing either of the message. This again can be well predicted by the analysis which indicates that we can expect an SNR gain of
\begin{equation*}
    \frac{1}{4} 10\log_{10}\left(\frac{2042}{41}\right) + \frac{1}{2} 10\log_{10}\left( 49 \right) \approx 12.69~\text{dB},
\end{equation*}
where $2042$ and $41$ inside the first logarithm are $N_{\mc{C}}$ and $N_{\mc{C}_{\mc{S}}}$, respectively and the $49$ inside the second logarithm is the ratio of $\delta(\mc{C}_{\mc{S}})$ and $\delta(\mc{C})$.

\begin{figure}
    \centering
    \includegraphics[width=4.5in]{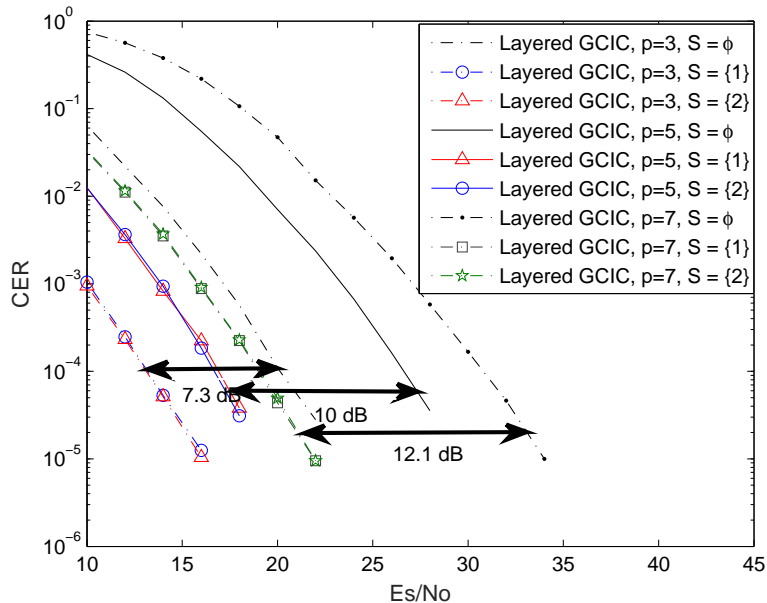}
    \caption{CER performance for the proposed layered Golden-coded index coding.}
    \label{fig:CER_GCIC_q3_q5}
\end{figure}

\begin{remark}
We end this section by showing that the proposed layered Golden-coded index coding is not a special case of the Golden-coded index coding in \cite{huang17GCIC} and vice versa. The Golden-coded index coding in \cite{huang17GCIC} is constructed over $\Z[\msf{e}][\theta]$ with ideals of the form $(\alpha+\beta\msf{e})\Z[\msf{e}][\theta]$ where $\alpha,\beta\in\Zi$. Consider $p=17$ for which \cite[Example 6]{huang17GCIC} indicates that $17\Z[\msf{e}][\theta]$ can be partitioned into 4 ideals, each with norm $17^2$. So the Golden-coded index coding can take messages of size $17^2$. To do the same for our layered scheme, it requires an ideal in $\Z[\msf{i}][\theta]$ to have norm $17$, which is impossible from the result in Table~\ref{tbl:GC}. Now, let us consider $p=29$ where Table~\ref{tbl:GC} shows that $29\Z[\msf{i}][\theta]$ can be partitioned into four ideals, each with norm $29$. Hence, the proposed layered Golden-coded index coding can take messages of size $29^2$. This will require $29\Z[\msf{e}][\theta]$ to be partitioned into ideals of the form $\alpha+\beta\msf{e}$ with norm $29^2$. However, using Magma, we obtain that $29\Z[\msf{e}][\theta]=\mfk{I}_1\mfk{I}_2\mfk{I}_3\mfk{I}_4$ with $\mfk{I}_1=(\bar{\theta}+2\msf{i})\Z[\msf{e}][\theta]$, $\mfk{I}_2=(\bar{\theta}-2\msf{i})\Z[\msf{e}][\theta]$, $\mfk{I}_3=(\theta+2\msf{i})\Z[\msf{e}][\theta]$, and $\mfk{I}_4=(\theta-2\msf{i})\Z[\msf{e}][\theta]$, where none of these satisfies the form required by the Golden-coded index coding.
\end{remark}

\section{LSTIC based on $3\times 3$ perfect STBC}\label{sec:3by3}
Let $\zeta_{7}$ be the 7th root of unity and let $\theta \defeq \zeta_7 + \zeta_7^{-1} = 2\cos\left(\frac{2\pi}{7}\right)$. Also, let $\KK=\Q(\omega)$ and let $\LL=\Q(\omega,\theta)$ the field extension of $\KK$ with $[\LL:\KK]=3$. Consider the cyclic division algebra
\begin{equation*}
    \mc{A} = (\LL/\KK,\sigma,\gamma) = \{x_0 + x_1\msf{e} + x_2\msf{e}^2 | x_0,\ldots,x_2\in\LL\},
\end{equation*}
where $\sigma:\zeta_7+\zeta_7^{-1}\rightarrow \zeta_7^2+\zeta_7^{-2}$ and $\msf{e}^3 = \gamma \defeq j$. A $3\times 3$ perfect STBC is constructed from
\begin{equation*}
    \bar{\mc{A}}_{\mfk{I}} = \{\alpha x_0 + \alpha x_1\msf{e} + \alpha x_2\msf{e}^2 | x_0,\ldots,x_2\in\OL\},
\end{equation*}
where $\alpha = 1+\omega+\theta$. The code will have the matrix form shown in \eqref{eqn:SP_code_matrix_principal}.

One can now follow Section~\ref{subsec:LSTIC_principal} to construct LSTIC based on $3\times 3$ perfect STBC. As a result, we have the following corollary whose proof is identical to that of Theorem~\ref{thm:SI_gain_principal} together with the fact that $\OL=\Zw[\theta]$ is a principal ideal domain.
\begin{corollary}
    For any $\mc{S}\subset\{1,\ldots,K\}$, the proposed LSTIC based on $3\times 3$ perfect STBC provides $\Gamma(\mc{C},\mc{S}) = 6$ dB/bits per real symbol.
\end{corollary}

\subsection{Examples and Simulation Results}
Here, we again factorize each prime $p<100$ into prime ideals via Magma. We show ideals and their inertial degrees $f$. The ramification index of each prime ideal lying above $p$ is given by
\begin{equation*}
    e = \left\{
            \begin{array}{ll}
              2, & \hbox{$p = 3$;} \\
              3, & \hbox{$p=7$;} \\
              1, & \hbox{otherwise.}
            \end{array}
          \right.
\end{equation*}
This can be justified by observing that
\begin{equation*}
    d_{\LL} =  3^3 7^4,
\end{equation*}
%\begin{equation}
%    \Delta_{\LL/\mbb{Q}} = N_{\KK/\mbb{Q}}(\Delta_{\LL/\KK})\cdot\Delta_{\KK/\mbb{Q}}^{[\LL:\KK]} = 3^3 7^4,
%\end{equation}
which has prime factors 3 and 7. Again, since $\OL$ is a principal ideal domain, every $p\OL$ can be factorized into principal prime ideals as shown in Table~\ref{tbl:LSTIC3}.
\begin{table}
\centering
\caption{Prime factorization of $p<100$ in $\mbb{Z}[\omega][\theta]$ where $\theta=\zeta_{7} + \zeta_{7}^{-1}$.}
\label{tbl:LSTIC3}
\begin{tabular}{|c|c|c|c|}
  \hline
  $p$ & $\mfk{p}$ & $f$ \\
\hline
  2 & $(2)$ & 6\\
  3 & $(1+\omega)$ & 3\\
  5 & $(5)$ & 6\\
  7 & $((\omega - 1)\theta^2 + (\omega - 1)\theta - \omega + 2), ((-\omega + 1)\theta^2 - (\omega - 1)\theta + 2\omega - 1)$ & 1\\
  11 & $(11)$ & 6\\
  13 & $(\omega\theta^2 + (\omega - 1)\theta - \omega - 1), ((\omega - 1)\theta^2 - \theta - \omega + 1), (-\theta^2 - \omega\theta + 2)$ & 1 \\
     & $ (\omega\theta^2 + \theta - 2\omega + 1), (-\omega\theta^2 - \theta + \omega), (\omega\theta^2 + \theta - 2\omega) $ & \\
  17 & $(17)$ & 6\\
  19 & $(3-5\omega), (3\omega - 5)$ & 3\\
  23 & $(23)$ & 6\\
  29 & $((2\omega - 2)\theta^2 - (\omega - 1)\theta - 4\omega + 4), (3\omega\theta^2 + 2\omega\theta - 4\omega), (3\omega\theta^2 + \omega\theta - 4\omega)$ & 2\\
  31 & $(\omega + 5), (5\omega + 1)$ & 3\\
  37 & $(7\omega - 4), (3\omega + 4)$ & 3\\
  41 & $((\omega - 1)\theta^2 - (2\omega - 2)\theta - 4\omega + 4), (3\theta^2 + \theta - 3), ((2-2\omega )\theta^2 - (3\omega - 3)\theta + 4\omega - 4)$ & 2\\
  43 & $((\omega - 1)\theta^2 + \theta - 2\omega + 2), (\theta^2 + (-\omega + 2)\theta - 1), ((\omega + 1)\theta^2 + \theta - 2\omega - 1)$ & 1\\
     & $(\theta^2 + (\omega - 1)\theta - 2), (-\theta^2 + (-\omega - 1)\theta + 1), ((-\omega + 2)\theta^2 + \theta + 2\omega - 3)$ & \\
  47 & $(47)$ & 6\\
  53 & $(53)$ & 6\\
  59 & $(59)$ & 6\\
  61 & $(5\omega + 4), (4\omega + 5)$ & 3\\
  67 & $(7-9\omega), (7\omega - 9)$ & 3\\
  71 & $(\theta^2 + \theta + 3), (4\theta^2 + 3\theta - 5), ((\omega - 1)\theta^2 - 6\omega + 6)$ & 2\\
  73 & $(8\omega - 9), (9-\omega)$ & 3\\
  79 & $(7\omega + 3), (3\omega + 7)) $& 3\\
  83 & $(2\theta^2 - 2\theta - 5), (4\theta^2 + 2\theta - 5), (2\omega\theta^2 + 4j\theta - 3\omega)$ & 2\\
  89 & $(89)$ & 6\\
  97 & $(-\theta^2 - \theta - 2\omega + 3), (\theta^2 - 2\omega), (\theta - 2\omega + 2)$ & 1\\
     & $ (-\theta^2 - \theta + 2\omega + 1), (\theta^2 + 2\omega - 2), (\theta + 2\omega)$ & \\
  \hline
\end{tabular}
\end{table}

Simulation results for the $3\times 3$ case are presented in Fig.~\ref{fig:CER_LSTIC3_q7} where we construct LSTIC from the $3\times 3$ perfect STBC with two principal ideals generated by $\beta_1=((\omega - 1)\theta^2 + (\omega - 1)\theta - \omega + 2)$ and  $\beta_2=((-\omega + 1)\theta^2 - (\omega - 1)\theta + 2\omega - 1)$. From Table~\ref{tbl:LSTIC3}, we learn that both $\beta_1$ and $\beta_2$ correspond to $p=7$ and we have $\beta_1\beta_2\OL=7\OL$. Hence, the overall codebook corresponds to \eqref{eqn:SP_code_matrix_principal} with $x_0,x_1,x_2\in\OL/3\OL$. Fig.~\ref{fig:CER_LSTIC3_q7} indicates that by revealing either of the message to the receiver, one obtains a roughly 10.5 dB SNR reduction. On the other hand, our analysis shows that the SNR reduction one can expect is roughly
\begin{equation*}
    \frac{1}{9} 10\log_{10}\left(\frac{5.9\times 10^{10}}{652428}\right) + \frac{1}{3} 10\log_{10}\left( 343 \right) \approx 13.95~\text{dB},
\end{equation*}
where the parameters inside the first and second logarithms are corresponding to gains in $N_{\mc{C}}$ and $\delta(\mc{C})$, respectively. The difference between the simulation results and our analysis is largely due to the fact that the SNR gain is measured at $10^{-4}$ CER, which is far from the asymptotic regime for a $3\times 3$ STBC. This is evident by observing that the CER curves have not even exhibited the promised diversity order of 9.

\begin{figure}
    \centering
    \includegraphics[width=4.5in]{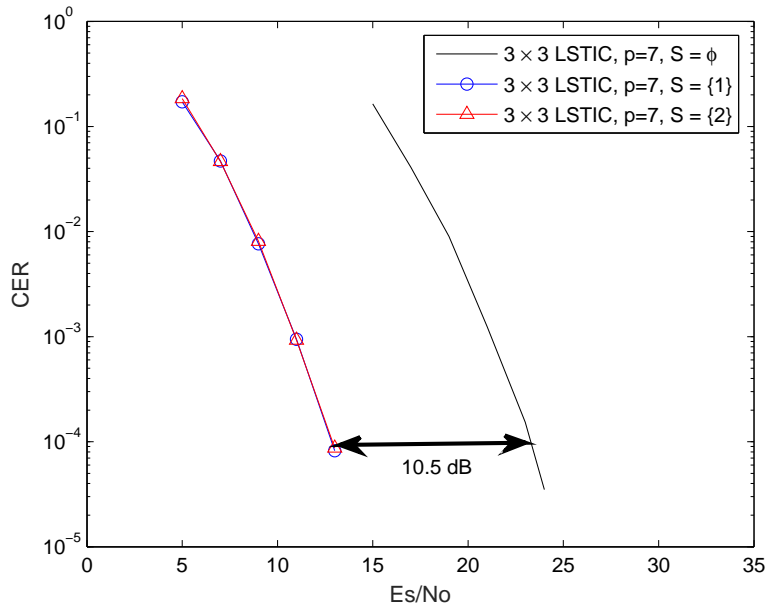}
    \caption{CER performance for the proposed LSTIC constructed from $3\times 3$ STBC.}
    \label{fig:CER_LSTIC3_q7}
\end{figure}

\section{LSTIC based on $4\times 4$ perfect STBC}\label{sec:4by4}
Let $\zeta_{15}$ be the 15th root of unity and let $\theta \defeq \zeta_{15} + \zeta_{15}^{-1} = 2\cos\left(\frac{2\pi}{15}\right)$. Also, let $\KK=\Q(\msf{i})$ and let $\LL=\Q(\msf{i},\theta)$ the field extension of $\KK$ with $[\LL:\KK]=4$. Consider the cyclic division algebra
\begin{equation*}
    \mc{A} = (\LL/\KK,\sigma,\gamma) = \{x_0 + x_1\msf{e} + x_2\msf{e}^2 + x_3\msf{e}^3 | x_0,\ldots,x_3\in\LL\},
\end{equation*}
where $\sigma:\zeta_{15}+\zeta_{15}^{-1}\rightarrow \zeta_{15}^2+\zeta_{15}^{-2}$ and $\msf{e}^4 = \gamma \defeq i$. A $4\times 4$ perfect STBC is constructed from
\begin{equation*}
    \bar{\mc{A}}_{\mfk{I}} = \{\alpha x_0 + \alpha x_1\msf{e} + \alpha x_2\msf{e}^2 + \alpha x_3\msf{e}^3 | x_0,\ldots,x_3\in\OL\},
\end{equation*}
where $\alpha = (1-3\msf{i})+\msf{i}\theta^2$. The code will have the matrix form shown in \eqref{eqn:SP_code_matrix_principal}.

One can now follow Section~\ref{subsec:LSTIC_principal} to construct LSTIC based on $4\times 4$ perfect STBC. As a result, we have the following corollary.
\begin{corollary}
    For any $\mc{S}\subset\{1,\ldots,K\}$, the proposed LSTIC based on $4\times 4$ perfect STBC provides $\Gamma(\mc{C},\mc{S}) \geq 6$ dB/bits per real symbol. Moreover, if all $\mfk{q}_k$, $k\in\{1,\ldots,K\}$, are principal, then $\Gamma(\mc{C},\mc{S}) = 6$ dB/bits per real symbol.
\end{corollary}

\subsection{Examples and Simulation Results}
Here, we factorize each prime $p<100$ into prime ideals via Magma. In Table~\ref{tbl:LSTIC4}, we show ideals and their inertial degrees $f$. The ramification index of each prime ideal lying above $p$ is given by
\begin{equation*}
    e = \left\{
            \begin{array}{ll}
              2, & \hbox{$p = 2, 3$;} \\
              4, & \hbox{$p=5$;} \\
              1, & \hbox{otherwise.}
            \end{array}
          \right.
\end{equation*}
This can be justified by observing that
\begin{equation*}
    d_{\LL} = 2^8 3^4 5^6,
\end{equation*}
%\begin{equation}
%    \Delta_{\LL/\mbb{Q}} = N_{\KK/\mbb{Q}}(\Delta_{\LL/\KK})\cdot\Delta_{\KK/\mbb{Q}}^{[\LL:\KK]} = 2^8 3^4 5^6,
%\end{equation}
which has prime factors 2, 3, and 5. Also, note that in this case, $p=3, 5, 29, 89$ are factorized into non-principal prime ideals.
\begin{table}
\centering
\caption{Prime factorization of $p<100$ in $\mbb{Z}[\msf{i}][\theta]$ where $\theta=\zeta_{15} + \zeta_{15}^{-1}$. For $p=29, 89$, we only list one of the eight ideals due to the space limitation; the other seven ideals can be obtained as the conjugates.}
\label{tbl:LSTIC4}
\begin{tabular}{|c|c|c|c|}
  \hline
  $p$ & $\mfk{p}$ & $f$\\
\hline
  2 & $(1+\msf{i})$ & 4 \\
  3 & $(3, (5\msf{i} + 2)\theta^3 + 7\msf{i}\theta^2 + (4\msf{i} + 4)\theta + 7\msf{i} + 7), (3, 2\theta^3 + 2\msf{i}\theta^2 + (7\msf{i} + 5)\theta + 8\msf{i} + 6)$ & 2\\
  5 & $(5, (14\msf{i} + 21)\theta^3 + (10\msf{i} + 1)\theta^2 + (12\msf{i} + 21)\theta + 4\msf{i} + 22)$ & 1\\
    & $(5, (20\msf{i} + 7)\theta^3 + (6\msf{i} + 9)\theta^2 + (21\msf{i} + 8)\theta + 8\msf{i} + 5)$ & \\
  7 & $((-\msf{i} + 1)\theta^3 + (3\msf{i} - 3)\theta - 2\msf{i} - 1), ((\msf{i} + 1)\theta^3 - (3\msf{i} + 3)\theta + 2\msf{i} - 1)$ & 4 \\
  11 & $(2\msf{i}\theta^3 + \msf{i}\theta^2 - 6\msf{i}\theta - \msf{i} + 1), (\msf{i}\theta^3 - 2\msf{i}\theta + 1)$ & 2 \\
     & $((\msf{i} + 1)\theta^3 + \msf{i}\theta^2 - (3\msf{i} + 4)\theta - \msf{i}), (\theta^3 - \theta^2 - 3\theta - \msf{i} + 3)$ & \\
  13 & $(2+3\msf{i}), (2-3\msf{i})$ & 4 \\
  17 & $(4+\msf{i}), (4-\msf{i})$ & 4 \\
  19 & $(i\theta^3 - (4\msf{i} - 1)\theta + \msf{i}), (\msf{i}\theta^3 - (1-\msf{i})\theta^2 - 3\msf{i}\theta - 2\msf{i} + 2)$& 2  \\
     & $(i\theta^3 + (\msf{i} + 1)\theta^2 - 3\msf{i}\theta - 2\msf{i} - 2), (\msf{i}\theta^3 - (4\msf{i} + 1)\theta +\msf{i})$  & \\
  23 & $((3\msf{i} + 3)\theta^3 - (9\msf{i} + 9)\theta + 2\msf{i} + 1), ((3\msf{i} - 3)\theta^3 - (9\msf{i} - 9)\theta + 2\msf{i} - 1)$ & 4  \\
  29 & $(29,(813\msf{i} + 779)\theta^3 + (812\msf{i} + 793)\theta^2 + (755\msf{i} + 41)\theta + 814\msf{i} + 5)$ & 1  \\
  31 & $(-2\msf{i}\theta^2 + 5\msf{i}), (2\msf{i}\theta^3 + 2\msf{i}\theta^2 - 6\msf{i}\theta - 3\msf{i}), (2\theta^3 - 8\theta + 1), (2\theta - 1)$ & 2  \\
  37 & $(6+\msf{i}), (6-\msf{i})$ & 4  \\
  41 & $((\msf{i} + 2)\theta^3 - (3\msf{i} + 6)\theta + \msf{i} + 1), ((2\msf{i} - 1)\theta^3 + (3-6\msf{i})\theta + \msf{i})$ & 2  \\
     & $ ((2\msf{i} + 1)\theta^3 - (6\msf{i} + 3)\theta + \msf{i} + 1), ((2\msf{i} + 1)\theta^3 - (6\msf{i} + 3)\theta + \msf{i})$  & \\
  43 & $((\msf{i} - 1)\theta^3 - (3\msf{i} - 3)\theta + 5\msf{i} + 4),((\msf{i} + 1)\theta^3 - (3\msf{i} + 3)\theta + 5\msf{i} -4)$ & 4 \\
  47 & $((3\msf{i} + 3)\theta^3 - (9\msf{i} + 9)\theta + 5\msf{i} - 2),((3\msf{i} - 3)\theta^3 - (9\msf{i} - 9)\theta + 5\msf{i} + 2)$ & 4\\
  53 & $(2+7\msf{i}), (2-7\msf{i})$ & 4 \\
  59 & $(-2\theta^3 - \theta^2 + 7\theta - 1), (-\msf{i}\theta^3 - \msf{i}\theta^2 + 2\msf{i}\theta + 4\msf{i}), (-\msf{i}\theta^2 + \msf{i}\theta + 4\msf{i}), (-\theta^3 + \theta^2 + 4\theta - 1)$ & 2 \\
  61 & $(\theta + \msf{i}), (\theta^3 - 4\theta + \msf{i} + 1), (-\theta^3 - \theta^2 + 3\theta + \msf{i} + 2), (\theta^2 + \msf{i} - 2)$ & 1 \\
     & $ (\theta - \msf{i}), (\theta^3 - 4\theta - \msf{i} + 1), (-\theta^3 - \theta^2 + 3\theta - \msf{i} + 2), (\theta^2 - \msf{i} - 2)$  & \\
  67 & $((5\msf{i} + 5)\theta^3 - (15\msf{i} + 15)\theta + 4\msf{i} + 1),((5\msf{i} - 5)\theta^3 - (15\msf{i} - 15)\theta + 4\msf{i} - 1)$ & 4 \\
  71 & $((3\msf{i} + 1)\theta^3 + (\msf{i} + 1)\theta^2 - (10\msf{i} + 3)\theta - \msf{i}), ((\msf{i} - 3)\theta^3 - 2\theta^2 + (9-4\msf{i})\theta + \msf{i} + 2)$ & 2 \\
     & $ ((\msf{i} + 3)\theta^3 + 2\theta^2 - (4\msf{i} + 9)\theta + \msf{i} - 2), ((3\msf{i} - 1)\theta^3 + (\msf{i} - 1)\theta^2 - (10\msf{i} - 3)\theta - \msf{i})$   &\\
  73 & $(3+8\msf{i}), (3-8\msf{i})$ & 4 \\
  79 & $((2\msf{i} + 1)\theta^3 - (7\msf{i} + 2)\theta + \msf{i}), ((\msf{i} + 2)\theta^3 - (\msf{i} - 1)\theta^2 - (3\msf{i} + 6)\theta + 3\msf{i} - 1)$ & 2 \\
     &  $ ((\msf{i} - 2)\theta^3 - (\msf{i} + 1)\theta^2 - (3\msf{i} - 6)\theta + 3\msf{i} + 1), ((1-2\msf{i})\theta^3 + (7\msf{i} - 2)\theta - \msf{i})$ &  \\
  83 & $((3\msf{i} - 3)\theta^3 - (9\msf{i} - 9)\theta + 7\msf{i} + 4), ((3\msf{i} + 3)\theta^3 - (9\msf{i} + 9)\theta + 7\msf{i} - 4)$ & 4 \\
  89 & $(89,(27\msf{i} + 82)\theta^3 + (31\msf{i} + 117)\theta^2 + (77\msf{i} + 7669)\theta + 7896\msf{i} + 7771)$ & 1 \\
  97 & $(9+4\msf{i}),(9-4\msf{i})$ & 4 \\
  \hline
\end{tabular}
\end{table}

\begin{figure}
    \centering
    \includegraphics[width=4.5in]{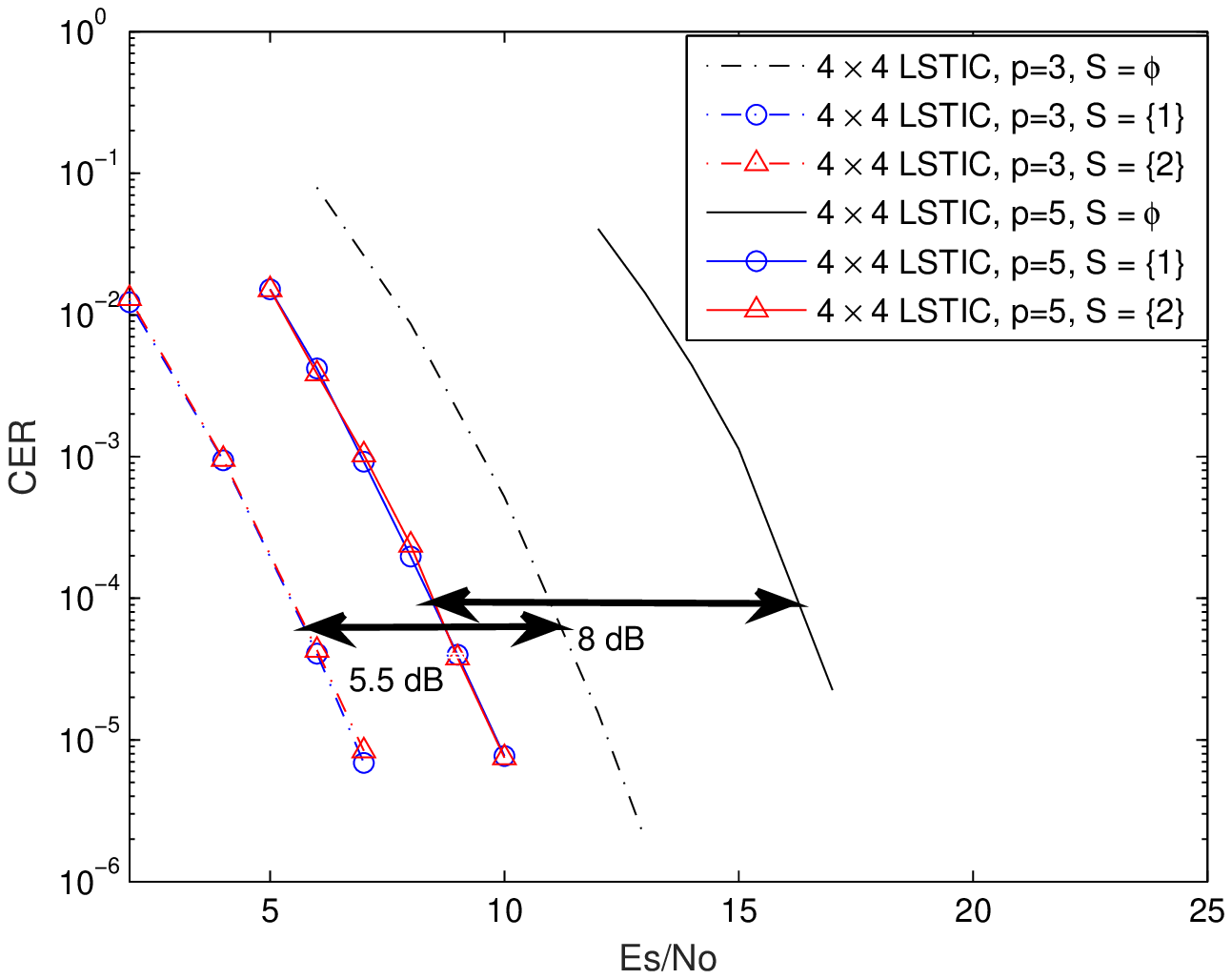}
    \caption{CER performance for the proposed LSTIC constructed from $4\times 4$ STBC.}
    \label{fig:CER_LSTIC4_q3_q5}
\end{figure}

In Fig.~\ref{fig:CER_LSTIC4_q3_q5}, two sets of simulation results are presented. Let us consider ideals $\mfk{I}_1 = (3, (5\msf{i} + 2)\theta^3 + 7\msf{i}\theta^2 + (4\msf{i} + 4)\theta + 7\msf{i} + 7)$ and $\mfk{I_2} = (3, 2\theta^3 + 2\msf{i}\theta^2 + (7\msf{i} + 5)\theta + 8\msf{i} + 6)$. From Table~\ref{tbl:LSTIC4} and the ramification index of $3$, we learn that $3\OL = \mfk{I}_1^2\mfk{I}_2^2$ where $N(\mfk{I}_1^2)=N(\mfk{I}_2^2)=81$. Moreover, with some computation, we have that $\mfk{I}_1^2$ and $\mfk{I}_2^2$ are principal ideals with generators $\beta_1=(\msf{i} + 1)\theta^3 - 3(\msf{i} + 1)\theta + 1$ and $\beta_2 = (\msf{i} - 1)\theta^3 - 3(\msf{i} - 1)\theta - 1$, respectively. In the first set, we construct LSTIC from $4\times 4$ perfect STBC with two principal ideals corresponding to $p=3$ generated by $\beta_1$ and $\beta_2$, respectively. Each message consists of four sub-messages from $\mbb{Z}_{81}$ and the overall codebook corresponds to the one in \eqref{eqn:SP_code_matrix_principal} with $x_0,x_1,x_2,x_3\in\OL/3\OL$. Fig.~\ref{fig:CER_LSTIC4_q3_q5} indicates a roughly 5.5 dB SNR gain by revealing either message to the receiver. We note that the analysis predicts a roughly
\begin{equation*}
    \frac{1}{16} 10\log_{10}\left(\frac{4.89\times 10^9}{9099}\right) + \frac{1}{4} 10\log_{10}\left( 81 \right) \approx 8.35~\text{dB},
\end{equation*}
where again the parameters inside the first and second logarithms are corresponding to gains in $N_{\mc{C}}$ and $\delta(\mc{C})$, respectively.

In the second set of simulations, we consider ideals $\mfk{I}_1=(5, (14\msf{i} + 21)\theta^3 + (10\msf{i} + 1)\theta^2 + (12\msf{i} + 21)\theta + 4\msf{i} + 22)$ and $\mfk{I}_2=(5, (20\msf{i} + 7)\theta^3 + (6\msf{i} + 9)\theta^2 + (21\msf{i} + 8)\theta + 8\msf{i} + 5)$ that correspond to $p=5$. Again from~\ref{tbl:LSTIC4} and the ramification index of $5$, we learn that $5\OL = \mfk{I}_1^4\mfk{I}_2^4$ where $N(\mfk{I}_1^2)=N(\mfk{I}_2^2)=625$. We have that $\mfk{I}_1^4$ and $\mfk{I}_2^4$ are principal ideals with generators $\beta_1=2\msf{i}-1$ and $\beta_2 = 2\msf{i}+1$, respectively. We again construct LSTIC from $4\times 4$ perfect STBC with two principal ideals generated by $\beta_1$ and $\beta_2$, respectively. Simulation result in Fig.~\ref{fig:CER_LSTIC4_q3_q5} shows a roughly 8 dB SNR gain obtained by revealing either message to the receiver. We again note that the analysis predicts a SNR gain of roughly
\begin{equation*}
    \frac{1}{16} 10\log_{10}\left(\frac{4.65\times 10^{14}}{2.18\times 10^6}\right) + \frac{1}{4} 10\log_{10}\left( 625 \right) \approx 12.19~\text{dB}.
\end{equation*}
In both the cases, one observes that there is a difference between the simulation results and the analysis. This again can be explained by that the CER where we measure the side information gain is far from the asymptotic regime for a $4\times 4$ STBC, which is again evident by observing that the CER curves have not exhibited the promised diversity order of 16.

\section{LSTIC based on $6\times 6$ perfect STBC}\label{sec:6by6}
Let $\zeta_{28}$ be the 28th root of unity and let $\theta \defeq \zeta_{28} + \zeta_{28}^{-1} = 2\cos\left(\frac{\pi}{14}\right)$. Also, let $\KK=\Q(\omega)$ and let $\LL=\Q(\omega,\theta)$ the field extension of $\KK$ with $[\LL:\KK]=6$. Consider the cyclic division algebra
\begin{equation*}
    \mc{A} = (\LL/\KK,\sigma,\gamma) = \{x_0 + x_1\msf{e}+\ldots + x_5\msf{e}^5 | x_0,\ldots,x_5\in\LL\},
\end{equation*}
where $\sigma:\zeta_{28}+\zeta_{28}^{-1}\rightarrow \zeta_{28}^5+\zeta_{28}^{-5}$ and $\msf{e}^6 = \gamma \defeq -\omega$. A $6\times 6$ perfect STBC is constructed from
\begin{equation*}
    \bar{\mc{A}}_{\mfk{I}} = \{x_0 + x_1\msf{e}+\ldots + x_5\msf{e}^5 | x_0,\ldots,x_5\in\mfk{I}\},
\end{equation*}
where $\mfk{I}$ is such that $7\OL = \mfk{I}^6\bar{\mfk{I}}^6$.

One can now follow Section~\ref{subsec:LSTIC_nonprincipal} to construct LSTIC based on $6\times 6$ perfect STBC. As a result, we have the following corollary.
\begin{corollary}
    For any $\mc{S}\subset\{1,\ldots,K\}$, the side information gain achieved by the proposed LSTIC based on $6\times 6$ perfect STBC with non-principal ideal $\mfk{I}$ is lower bounded as
    \begin{equation*}
        \Gamma(\mc{C},\mc{S}) \geq 6 + \gamma_{\mfk{I}}~\text{dB/bits per real symbol},
    \end{equation*}
    where $\gamma_{\mfk{I}}$ is as shown in \eqref{eqn:gamma_I}.
\end{corollary}

%\subsection{Examples and Simulation Results}
In Table~\ref{tbl:LSTIC6}, we again factorize each prime $p<100$ into prime ideals via Magma. We show ideals and their inertial degrees $f$. The ramification index of each prime ideal lying above $p$ is given by
\begin{equation*}
    e = \left\{
            \begin{array}{ll}
              2, & \hbox{$p = 2, 3$;} \\
              6, & \hbox{$p=7$;} \\
              1, & \hbox{otherwise.}
            \end{array}
          \right.
\end{equation*}
This can be justified by observing that
\begin{equation*}
    d_{\LL} = 2^{12} 3^6 7^{10},
\end{equation*}
which has prime factors 2, 3, and 7. In this case, for $p<100$, we note that $p=3, 7, 19, 31$ are factorized into non-principal prime ideals.
\begin{table}
\centering
\caption{Prime factorization of $p<100$ in $\mbb{Z}[\omega][\theta]$ where $\theta=\zeta_{28} + \zeta_{28}^{-1}$.}
\label{tbl:LSTIC6}
\begin{tabular}{|c|c|c|c|}
  \hline
  $p$ & $\mfk{p}$ & $f$ \\
\hline
  2 & $(\theta^4 - 5\theta^2 + \theta + 5)$ & 6 \\
  3 & $(3, (7\omega + 6)\theta^5 + (3\omega + 6)\theta^4 + (3\omega + 2)\theta^3 + (3\omega + 8)\theta^2 + (8\omega + 4)\theta + 3\omega + 5)$ & 3 \\
    & $(3, (7\omega + 8)\theta^5 + (\omega + 1)\theta^4 + (\omega + 4)\theta^3 + (7\omega + 7)\theta^2 + (3\omega + 7)\theta + 5\omega + 6)$ & \\
  5 & $((\omega - 1)\theta^5 + (\omega - 2)\theta^4 + (4-5\omega)\theta^3 + (9-5\omega)\theta^2 + (5\omega - 2)\theta + 6\omega - 8)$ & 6 \\
    & $((1-\omega)\theta^5 + (1-2\omega)\theta^4 + (4\omega - 5)\theta^3 + (9\omega - 5)\theta^2 + (5-2\omega)\theta - 8\omega + 6)$ & 6 \\
  7 & $(7, (32\omega + 15)\theta^5 + (22\omega + 21)\theta^4 + (23\omega + 14)\theta^3 + (44\omega + 10)\theta^2 + (18\omega + 21)\theta + 3\omega + 20)$ & 1 \\
    & $(7, (36\omega + 12)\theta^5 + (22\omega + 27)\theta^4 + (26\omega + 36)\theta^3 + (23\omega + 42)\theta^2 + (14\omega + 43)\theta + 9\omega + 29)$ & \\
  11 & $((\omega + 1)\theta^5 - (2\omega + 1)\theta^4 - (4\omega + 5)\theta^3 + (9\omega + 5)\theta^2 + (2\omega + 5)\theta - 8\omega - 6)$ & 6 \\
     & $((\omega + 1)\theta^5 + (2\omega + 1)\theta^4 - (4\omega + 5)\theta^3 - (9\omega + 5)\theta^2 +(2\omega + 5)\theta + 8+ 6) $ &  \\
  13 & $(\omega\theta^4 - (5\omega - 1)\theta^2 + 5\omega - 3), ((1-\omega)\theta^4 + (4\omega - 5)\theta^2 - 2\omega + 4), ((\omega - 1)\theta^4 + (5-4\omega )\theta^2 + 3\omega - 5)$ & 2 \\
     & $((\omega - 1)\theta^4 + (4-5\omega)\theta^2 + 5\omega - 2), (\omega\theta^4 - (4\omega + 1)\theta^2 + 2\omega + 2), ((1-\omega)\theta^4 + (5\omega - 4)\theta^2 - 5\omega + 3)$ &  \\
  17 & $((\omega - 1)\theta^4 - (2\omega + 1)\theta^3 + (5-2\omega)\theta^2 + (3\omega + 3)\theta + 2\omega - 5)$ & \\
     & $((\omega - 1)\theta^5 + 3\theta^4 + (3-6\omega)\theta^3 + (\omega - 13)\theta^2 + (8\omega +1)\theta - 4\omega + 10) $ & 6 \\
  19 & $(19, (61\omega + 187)\theta^5 + (107\omega + 256)\theta^4 + (123\omega + 152)\theta^3 + (87\omega + 168)\theta^2 + (100\omega + 76)\theta + 172\omega + 278)$ & 3 \\
     & $(19, (89\omega + 144)\theta^5 + 176\omega\theta^4 + (198\omega + 167)\theta^3 + (42\omega + 90)\theta^2 + (42\omega + 214)\theta + 134\omega + 293)$ & \\
     & $(19, (103\omega + 89)\theta^5 + (27\omega + 254)\theta^4 + (229\omega + 360)\theta^3 + (296\omega + 260)\theta^2 + (100\omega + 197)\theta + 61\omega + 239)$ & \\
     & $(19, (158\omega + 258)\theta^5 + (98\omega + 84)\theta^4 + (119\omega + 3)\theta^3 + (159\omega + 249)\theta^2 + (234\omega + 184)\theta + 28\omega + 19)$ & \\
  23 & $((3-2\omega)\theta^4 - (\omega + 1)\theta^3 + (9\omega - 16)\theta^2 + (\omega + 6)\theta - 9\omega + 16)$ & \\
     & $((3-\omega)\theta^4 + (2\omega - 1)\theta^3 + (7\omega - 16)\theta^2 + (6-7\omega)\theta - 7\omega + 16)$  & 6 \\
  29 & $(\theta^5 - \theta^4 - 5\theta^3 + 4\theta^2 + 5\theta - 3), (\theta^4 - \theta^3 - 5\theta^2 + 3\theta + 4), (\omega\theta^2 - \omega\theta - 3\omega)$ & 2 \\
     & $(\omega\theta^2 + \omega\theta - 3\omega), (\theta^4 + \theta^3 - 5\theta^2 - 3\theta + 4), (\theta^5 + \theta^4 - 5\theta^3 - 4\theta^2 + 5\theta + 3)$ & \\
  31 & $(31, (724\omega + 833)\theta^5 + (545\omega + 827)\theta^4 + (656\omega + 170)\theta^3 + (771\omega + 171)\theta^2 + (715\omega + 907)\theta + 680\omega + 916)$ & 3 \\
     & $(31, (45\omega + 21)\theta^5 + (266\omega + 398)\theta^4 + (942\omega + 59)\theta^3 + (506\omega + 472)\theta^2 + (43\omega + 69)\theta + 210\omega + 417)$ & \\
     & $(31, (927\omega + 236)\theta^5 + (56\omega + 700)\theta^4 + (151\omega + 808)\theta^3 + (525\omega + 9)\theta^2 + (749\omega + 157)\theta + 951\omega + 828)$ & \\
     & $(31, (194\omega + 143)\theta^5 + (848\omega + 7)\theta^4 + (305\omega + 255)\theta^3 + (521\omega + 168)\theta^2 + (378\omega + 357)\theta + 861\omega + 890)$ & \\
  37 & $((1-\omega )(\theta^5 -\theta^4 -5\theta^3) + (5-6\omega)(\theta^2 +\theta) + 5\omega - 3), ((1-\omega)\theta^4 + (\omega - 2)\theta^3 + (4\omega - 3)\theta^2 + (5-2\omega)\theta - 4\omega + 3)$ & 3 \\
     & $((\omega - 1)(\theta^5 -\theta^4 -5\theta^3) + (5\omega - 6)(\theta^2 +\theta) - 3\omega + 5), (\omega\theta^4 - (\omega + 1)\theta^3 - (4\omega - 1)\theta^2 + (2\omega + 3)\theta + 4\omega - 1)$ & \\
  41 & $(\theta^4 - \theta^3 - (4-\omega)\theta^2 - (\omega - 3)\theta - \omega + 2), (\theta^5 + \omega\theta^4 - (\omega + 5)\theta^3 - (5\omega + 1)\theta^2 + (3\omega + 5)\theta + 6\omega + 2)$ & 2 \\
     & $(\omega\theta^4 + (\omega - 1)\theta^3 - 4\omega\theta^2 + (3-3\omega)\theta + 4\omega - 2), (\omega\theta^4 - \theta^3 - (4\omega + 1)\theta^2 + (4-\omega)\theta + 2\omega + 3)$ & \\
     & $(\theta^3 - (\omega + 1)\theta^2 - 2\theta + 3\omega + 1), (\omega\theta^5 + \omega\theta^4 - 5\omega\theta^3 + (1-5\omega)\theta^2 + (5\omega + 1)\theta + 5\omega - 2)$ & \\
  43 & $((2-\omega )\theta^4 + (5\omega - 9)\theta^2 - 5\omega + 7), ((1-\omega)\theta^4 + (6\omega - 5)\theta^2 - 7\omega + 5), (\theta^4 + (\omega - 4)\theta^2 - 2\omega + 2)$ & 2 \\
     & $((2\omega - 1)\theta^4 + (5-9\omega)\theta^2 + 7\omega - 5), ((1-\omega)\theta^4 + (5\omega - 6)\theta^2 - 5\omega + 7), ((\omega - 1)\theta^4 + (4-3\omega)\theta^2 - 2)$ & \\
  47 & $((1-\omega)(2\theta^5 -2\theta^4 -7\theta^3 + 5\theta^2 + \theta + 8), ((\omega - 1)(2\theta^5 + 2\theta^4 -7\theta^3 -5\theta^2 + \theta - 8))$ & 6 \\
  53 & $(3\theta^4 + 5\theta^3 - 17\theta^2 - 13\theta + 17), ((\omega - 1)(3\theta^5 -2\theta^4  -20\theta^3 + 5\theta^2 + 30\theta + 8))$ & 6 \\
  59 & $(3\omega\theta^5 - 18\omega\theta^3 + 21\omega\theta + 2\omega), (3\omega\theta^5 - 18\omega\theta^3 + 21\omega\theta - 2\omega)$ & 6 \\
  61 & $(5-9\omega),(4-9\omega)$ & 6 \\
  67 & $(7\omega + 2),(2\omega + 7)$ & 6 \\
  71 & $((2-\omega)\theta^4 + (\omega - 1)\theta^3 + (5\omega - 9)\theta^2 + (4-3\omega)\theta - 5\omega+ 7), (-\theta^5 + \theta^4 + (6-\omega)\theta^3 - 6\theta^2 + (3\omega - 9)\theta - \omega + 8)$ & 2 \\
     & $(\theta^5 - \theta^4 - 5\theta^3 + (6-\omega)\theta^2 + (6-\omega)\theta + 2\omega - 7), (\theta^5 + \theta^4 - 5\theta^3 + (\omega - 6)\theta^2 + (6-\omega)\theta - 2\omega + 7)$ & \\
     & $((\omega - 1)(\theta^5 +\theta^4 -6\theta^2) + (5-6\omega)\theta^3 + (9\omega - 6)\theta + 8\omega - 7), (\omega\theta^4 - \theta^3 - 4\omega\theta^2 + 3\theta + 2\omega - 2)$ & \\
  73 & $(\omega + 8),(8\omega + 1)$ & 6 \\
  79 & $(7-10\omega),(3-10\omega)$ & 6 \\
  83 & $(\theta^5 - 5\theta^3 - \theta^2 + 6\theta + 4), ((2-2\omega)\theta^4 + (\omega - 1)\theta^3 + (9\omega - 9)\theta^2 + (2-2\omega)\theta - 8\omega + 8)$ & 2 \\
     & $((\omega - 1)(\theta^5 + \theta^4 -4\theta^3 -4\theta^2 + 2\theta + 4), (\theta^5 - 5\theta^3 + \theta^2 + 6\theta - 4)$ & \\
     & $(2\theta^4 + \theta^3 - 9\theta^2 - 2\theta + 8), (\theta^5 - \theta^4 - 4\theta^3 + 4\theta^2 + 2\theta - 4)$ & \\
  89 & $(5\omega\theta^4 - (2\omega + 1)\theta^3 + (1-18\omega)\theta^2 + (2-\omega)\theta + 18\omega -1)$) & 6 \\
     & $((3-2\omega)\theta^5 + (3-2\omega)\theta^4 + (17\omega - 23)\theta^3 + (15\omega -20)\theta^2 + (39-31\omega)\theta - 32\omega + 38)$ & 6 \\
  97 & $(\theta^2 + 2\omega - 4), (\omega\theta^4 - 4\omega\theta^2 + 2\omega + 2), ((1-\omega)\theta^4 + (5\omega - 5)\theta^2 - 3\omega + 5)$ & 2 \\
     & $(\theta^2 - 2\omega - 2), (\theta^4 - 4\theta^2 + 2\omega + 2), (\omega\theta^4 - 5\omega\theta^2 + 3\omega + 2)$ & \\
  \hline
\end{tabular}
\end{table}

\section{Layered Alamouti-Coded Index Coding}\label{sec:AlaIC}
In this section, we construct space-time index codes for $2\times 1$ MISO channel from Alamouti code \cite{alamouti98}. Alamouti code can be regarded as codes constructed over Hamilton quaternions  \cite{conway03quaternions}, the $\mbb{R}$-algebra of dimension 4 given by
\begin{align*}
    \mbb{H} = \{ a+b\mathsf{i}+c\mathsf{j}+d\mathsf{k} | a, b, c, d\in\mbb{R}\},
\end{align*}
where $\msf{i}^2=-1$, $\msf{j}^2=-1$, $\msf{k}^2=-1$, and $\msf{k}=\msf{ij}=-\msf{ji}$. We note that $\mbb{H}$ is a cyclic division algebra
\begin{equation*}
    \mbb{H} = (\Q(\msf{i})/\Q,\sigma,-1) = \{ x_0+\mathsf{j}x_1 | x_0, x_1\in\Q(i)\},
\end{equation*}
where $\sigma:\msf{i}\rightarrow -\msf{i}$ and $\lambda\msf{j} = \msf{j}\sigma(\lambda)$. This induces a layered structure of the Alamouti code. Now, consider $\bar{\mbb{H}}= \{ x_0+\mathsf{j}x_1 | x_0, x_1\in\Zi\}$, an Alamouti code corresponds to a finite subset of
\begin{equation*}%\label{eqn:Ala_inf}
    \mc{C}_{\Zi} \defeq \left\{ \left.\begin{pmatrix}
                                         x_0 & -x_1^* \\
                                         x_1 &  x_0^* \\
                                       \end{pmatrix}
                                       \right| x_0, x_1\in\Zi
     \right\}.
\end{equation*}
Thus, Alamouti code does not belong to the family of codes considered in Section~\ref{sec:LSTIC} (which have base fields $\mbb{K}=\Q(\msf{i})$ or $\Q(\msf{j})$). Fortunately, one can follow the same approach and obtain Alamouti-coded index coding as follows.
\begin{align*}
    \mbb{H} = \{ x_0+\mathsf{j}x_1 | x_0, x_1\in\mbb{C}\}.
\end{align*}
which  In what follows, we propose and analyze the layered Alamouti-coded index code using an approach similar to that in Section~\ref{sec:LSTIC}.

Note that $\Zi$ is a principal ideal domain; so every ideal can be generated by a singleton. Let $\phi_1,\ldots,\phi_K$ be $K$ elements in $\Zi$ that are relatively prime. Also, define $q=\Pi_{k=1}^K \phi_k$ and define $N(\phi_k)=q_k$ for $k\in\{1,\ldots,K\}$ where $q_k$s are not necessarily primes. From CRT, we have
\begin{equation*}
    \Zi/q\Zi \cong \Zi/\phi_1\Zi\times\ldots\times\Zi/\phi_K\Zi \cong \mbb{B}_{q_1} \times\ldots\times\mbb{B}_{q_K},
\end{equation*}
where $\mbb{B}_{q_k}=\Zi/\phi_k\Zi$ is a commutative ring with size $q_k$. Let $\mc{M}$ be an isomorphism that maps the messages onto a complete set of coset leaders of $\Zi/q\Zi$ with minimum energy. For $k\in\{1,\ldots,K\}$, we enforce $w_k\in\mbb{B}_{q_k}^2$ which can be represented as $w_k=(w_{k,0},w_{k,1})$ where each $w_{k,\ell}\in\mbb{B}_{q_k}$. The encoder maps $w_{1,\ell},\ldots,w_{K,\ell}$ into the signal of the layer $\ell\in\{0,1\}$ as
\begin{equation*}
    x_{\ell} = \mc{M}(w_{1,\ell},\ldots,w_{K,\ell})\in \Zi/q\Zi,\quad \ell\in\{0,1\}.
\end{equation*}
The overall codebook becomes a subset of $\mc{C}_{\Zi}$ given by
\begin{equation}\label{eqn:AlaIndOverall}
    \mc{C} \defeq \left\{ \left.\begin{pmatrix}
                                         x_0 & -x_1^* \\
                                         x_1 &  x_0^* \\
                                       \end{pmatrix}
                                       \right| x_0, x_1\in\Zi/q\Zi
     \right\}.
\end{equation}

For the proposed layered Alamouti-coded index coding, we provide the following result without proof. The proof is essentially identical to the proof of Theorem~\ref{thm:SI_gain_principal}.
\begin{theorem}\label{thm:SI_gain_ACIC}
    For any $\mc{S}\subset\{1,\ldots,K\}$, the proposed Alamouti-coded index coding provides $\Gamma(\mc{C},\mc{S}) = 6$ dB/bits per real symbol.
\end{theorem}

\subsection{Examples and Simulation Results}
Here, we list choices of $\phi_k$ lying above a prime $p<100$. In Table~\ref{tbl:ala}, we show principal ideals and their inertial degrees $f$. From $d_{\Q(\msf{i})}=4$, we know that the ramification index of each prime ideal lying above $p\neq 2$ is 1 and is 2 for prime ideals lying above $2$.
\begin{table}[h]
\centering
\caption{Prime factorization of $p<100$ in $\mbb{Z}[\msf{i}]$.}
\label{tbl:ala}
\begin{tabular}{|c|c|c|}
  \hline
  $p$ & $(\phi)$ & $f$\\
\hline
  2 & $(1+\msf{i})$ & 1 \\
  3 & (3) & 2 \\
  5 & $(1+2\msf{i}),(1-2\msf{i})$  & 1 \\
  7 & (7) & 2  \\
  11 & (11) & 2 \\
  13 & $(2+3\msf{i}),(2-3\msf{i})$ & 1 \\
  17 & $(1+4\msf{i}),(1-4\msf{i})$ & 1 \\
  19 & (19) & 2  \\
  23 & (23) & 2  \\
  29 & $(2+5\msf{i}),(2-5\msf{i})$ & 1 \\
  31 & (31) & 2  \\
  37 & $(1+6\msf{i}),(1-6\msf{i})$ & 1  \\
  41 & $(5+4\msf{i}),(5-4\msf{i})$ & 1  \\
  43 & (43) & 2 \\
  47 & (47) & 2 \\
  53 & $(2+7\msf{i}),(2-7\msf{i})$ & 1 \\
  59 & (59) & 2 \\
  61 & $(5+6\msf{i}),(5-6\msf{i})$ & 1 \\
  67 & (67) & 2 \\
  71 & (71) & 2 \\
  73 & $(3+8\msf{i}),(3-8\msf{i})$ & 1 \\
  79 & (79) & 2 \\
  83 & (83) & 2 \\
  89 & $(5+8\msf{i}),(5-8\msf{i})$ & 1 \\
  97 & $(4+9\msf{i}),(4-9\msf{i})$ & 1 \\
  \hline
\end{tabular}
\end{table}

\begin{figure}
    \centering
    \includegraphics[width=4.5in]{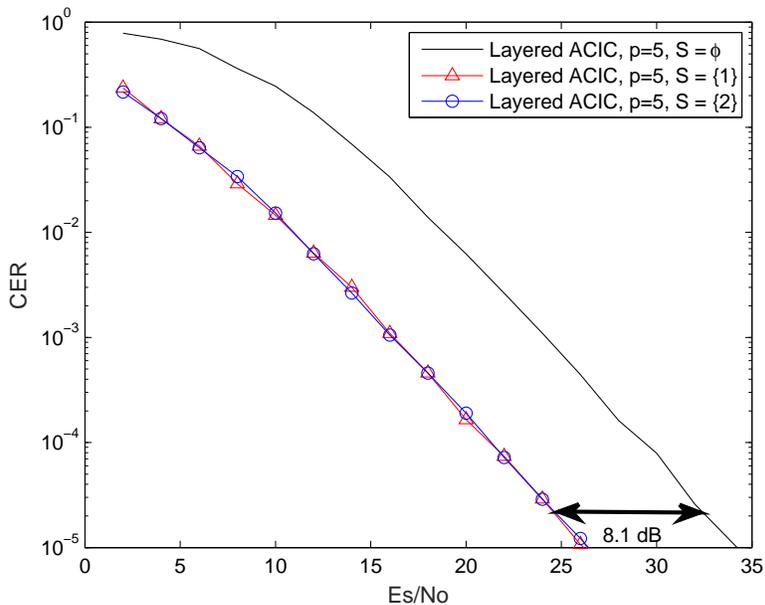}
    \caption{CER performance for the proposed layered Alamouti-coded index coding (ACIC).}
    \label{fig:CER_Ala_q5}
\end{figure}

Simulation results for using the proposed layered Alamouti-coded index coding over the $2\times 1$ MISO channel are provided in Fig.~\ref{fig:CER_Ala_q5}. In this figure, we construct the proposed layered Alamouti-index coding with two ideals generated by $\beta_1=1+2\msf{i}$ and $\beta_2=1-2\msf{i}$, respectively. From Table~\ref{tbl:ala}, we know that $5\Zi=\beta_1\beta_2\Zi$ and each ideal has norm equal to $p=5$. Each message consists of two sub-messages in $\Z_5$ and we encode the sub-messages of the same layer into the signal of that layer. The overall codebook becomes \eqref{eqn:AlaIndOverall} with $x_0,x_1\in\Zi/5\Zi$. The results in Fig.~\ref{fig:CER_Ala_q5} indicates a roughly 8.1 dB SNR gain when either message is revealed to the receiver. This can be accurately predicted by our analysis that revealing either message leads to an SNR gain given by
\begin{equation*}
    \frac{1}{2} 10\log_{10}\left(\frac{4}{2}\right) + \frac{1}{2} 10\log_{10}\left( 25 \right) \approx 8.49~\text{dB},
\end{equation*}
where $4$ and $2$ in the first logarithms are $N_{\mc{C}}$ and $N_{\mc{C}_{\mc{S}}}$, respectively, and $25$ inside the second logarithm corresponds to the gain in determinant.

\section{Conclusions}\label{sec:conclude}
In this paper, we have studied the problem of multicasting $K$ independent messages via MIMO links to multiple receivers where each of them already has a subset of messages as side information. A novel scheme, LSTIC, constructed over STBC has been proposed for exploiting side information without prior knowledge of the side information configuration. It has been shown that the proposed LSTIC possesses the nice property that for any possible side information the minimum determinant increases exponentially as the rate of the side information increases. Moreover, when constructed over perfect STBC, the perfect STBC properties are preserved by our construction and therefore the LSTIC is itself a perfect STBC. Examples including constructions of LSTIC over Golden code, $3\times 3$ perfect STBC, $4\times 4$ perfect STBC, $6\times 6$ perfect STBC, and Alamouti code have been provided and simulations have been conducted to corroborate our analysis.

\end{document}